\documentclass[a4paper,11pt,twoside]{article}
\usepackage{randolph}
\usepackage{libertine}
\usepackage{tikz}
\usepackage{thm-restate}
\usetikzlibrary{matrix}

\newcommand{\norm}[1]{  \lVert #1 \rVert}
\newcommand{\floor}[1]{  \lfloor #1 \rfloor}
\newcommand{\ceil}[1]{  \lceil #1 \rceil}
\newcommand{\ZZ}{\Z_{\geq 0}}

\newcommand{\DC}{\mathcal{C}}

\newcommand{\lelex}{\prec_\text{lex}}
\newcommand{\bigzero}{\mbox{\normalfont\Large\bfseries 0}}

\newcommand{\rvline}{\hspace*{-\arraycolsep}\vline\hspace*{-\arraycolsep}}

\DeclareUnicodeCharacter{00A0}{~}

\numberwithin{theorem}{section}
\numberwithin{lemma}{section}
\numberwithin{claim}{section}
\numberwithin{corollary}{section}
\numberwithin{definition}{section}
\numberwithin{observation}{section}
\numberwithin{remark}{section}
\newenvironment{claimproof}[1]{\par\noindent Proof:\space#1}{\hfill $\blacksquare$}

\title{ 
    Parameterized Algorithms on Integer Sets with Small Doubling: Integer Programming, Subset Sum and $k$-SUM
}

\author{
 Tim Randolph\\ 
 {Harvey Mudd College}\\
 \href{mailto:trandolph@hmc.edu}{\texttt{\small trandolph@hmc.edu}}
 \and 
 Karol Węgrzycki\\
 {Saarland University; Max Planck Institute for Informatics}\\
 \href{mailto:wegrzycki@cs.uni-saarland.de}{\texttt{\small wegrzycki@cs.uni-saarland.de}}
 \vspace{0.5em}
}
\date{}

\begin{document}

\maketitle
\thispagestyle{empty}

\begin{abstract}
    We study the parameterized complexity of algorithmic problems whose input is an integer set $A$
    in terms of the \emph{doubling constant} $\DC \coloneqq |A+A| / |A|$, a fundamental measure of additive structure. We present evidence that this new parameterization is algorithmically useful in the form of new results for two difficult, well-studied problems: Integer Programming and Subset Sum.

    First, we show that determining the feasibility of bounded Integer Programs is a tractable problem when parameterized in the doubling constant. Specifically, we prove that the feasibility of an integer program $\mathcal{I}$ with $n$ polynomially-bounded variables and $m$ constraints can be determined in time $n^{O_\DC(1)} \cdot \poly(|\mathcal{I}|)$ when the column set of the constraint matrix has doubling constant $\DC$.
    
    Second, we show that the Subset Sum and Unbounded Subset Sum problems can be solved in time $n^{O_C(1)}$ and $n^{O_\DC(\log\log\log n)}$, respectively, where the $O_C$ notation hides functions that depend only on the doubling constant $\DC$. We also show the equivalence of achieving an \FPT~algorithm for Subset Sum with bounded doubling and achieving a milestone result for the parameterized complexity of Box ILP. Finally, we design near-linear time algorithms for $k$-SUM as well as tight lower bounds for $4$-SUM and nearly tight lower bounds for $k$-SUM, under the $k$-SUM conjecture.

    Several of our results rely on a new proof that Freiman's Theorem, a central result in additive combinatorics, can be made efficiently constructive. This result may be of independent interest.
\end{abstract}

\clearpage
\setcounter{page}{1}
\section{Introduction}
\label{sec:intro}

Given a subset $X$ of a group, the \emph{doubling constant}
\begin{equation*}
    \DC \coloneqq \DC(X) = \frac{|X+X|}{|X|}
    \label{eq:DC}
\end{equation*}
is one measure used to capture the amount of ``additive structure'' in $X$. (Here, $X + X$ denotes the \emph{sumset} $\{a + b \; : \; a, b \in X\}$.) This value ranges (on integer sets) from $2 - o_n(1)$ for arithmetic progressions to $\frac{n}{2} + o_n(1)$ when all sums are distinct, and is central to the study of additive combinatorics. 
If the doubling constant $\DC$ is truly constant (that is, independent of the
set cardinality $|X|$), this indicates that $X$ is a ``highly structured'' set with respect to addition: for example, the statements that 
\begin{itemize}
        \setlength\itemsep{0.3em}
    \item $X$ has constant doubling ($|X+X| \leq c_1 |X|$), that
    \item the iterated sumset
    $sX \coloneqq \underbrace{X + X + X \dots + X}_{s \text{ times}}$
    is at most $c_2 \coloneqq c_2(s)$ times $|X|$, and that
    \item $X$ (likewise $X+X$ and $sX$) can be contained in a \emph{generalized arithmetic progression} of dimension $c_3$ and volume $c_4 |X|$,
\end{itemize}
are all equivalent up to the choice of constants $c_1$, $c_2$, $c_3$, and $c_4$ (c.f. \cite{tao2006additive} Proposition 2.26). The many fruitful applications of the doubling constant illustrate its significance as a robust measurement of additive structure (for an overview, see \cite{tao2006additive}, especially Chapter 2).


In this work, we consider the parameterized complexity of problems on integer sets with respect to the doubling constant. Specifically, we focus on two problems for which additive structure is particularly helpful: Integer Programming and Subset Sum.

\subsection{ Related Work }

Integer Programming and Subset Sum are not only problems in which additive structure plays an important role: they are also both well-studied and stubbornly difficult, to the point where significant work has gone into analysing their parameterized complexity and demarcating classes of tractable instances.

\subparagraph*{Integer Linear Programming.}
Many problems in \emph{combinatorial optimization} can be formulated as an \emph{integer linear program} (ILP). An ILP is an optimization problem of the
following form:

\begin{displaymath}
    \max \left\{ c^T x \mid A x = b, x \in \ZZ^n \right\},
\end{displaymath}

where $A \in \Z^{m \times n}$, $c \in \Z^n$ and $b \in \Z^m$. (ILPs of the form $Ax \geq b$ can be converted to this form using slack variables.) Unlike linear programming, integer programming is \textsf{NP}-complete.  Due
to its generality and both practical and theoretical importance, the
complexity of ILP has been rigorously studied through the lens of \emph{parameterized complexity}. Lenstra~\cite{lenstra83} has shown that an
integer linear program with a fixed number of variables can be solved in
polynomial time. His algorithm was subsequently improved, and the current record
is $(\log{n})^{O(n)}$~\cite{ILP-rothvoss}. The question of whether this can be brought
down to $2^{O(n)}$ is one of the most prominent open questions in the theory of
algorithms.

ILP can also be parameterized in the number of constraints $m$ and the maximum
absolute value of any coefficient in the constraint matrix, $\Delta \coloneqq \norm{A}_\infty$.
In 1981, Papadimitriou~\cite{Papadimitriou81} presented an $(m\Delta)^{O(m^2)}$-time algorithm, and the best algorithms for ILP parameterized in $m$ and $\Delta$ continue to improve: see~\cite{jansen2018integer,eisenbrand2019proximity} for recent progress.
Another class of tractable instances of ILP rely on structural properties of the constraint matrix (see~\cite{cslovjecsek2024parameterized,n-fold,DBLP:conf/soda/CslovjecsekEHRW21,DBLP:conf/esa/CslovjecsekEPVW21}).

\noindent
\subparagraph*{ Subset Sum. }
Along with the closely related Knapsack problem, the Subset Sum problem is the canonical \textsf{NP}-complete problem concerning addition in integer sets.  In
addition to \textsf{NP}-completeness, the problem appears difficult from the
standpoint of exact algorithms: despite significant attention (see, e.g.,
\cite{woeginger2008open,austrin2016dense,NederlofW21}), solving Subset Sum in
time $2^{(1/2 - c)n}$ for some constant $c>0$ remains a major open problem.
Except for ``log shaving'' results that improve runtime by subexponential factors
\cite{chen2023logshavingSS}, the exact runtime has not been improved in 50
years~\cite{horowitz1974computing}. The lack of progress in exact algorithms
motivates parameterized approaches, including a long line of pseudopolynomial-time
algorithms parameterized by the size of
the target~\cite{bringmann2017near,koiliaris2019faster,abboud2022seth} and the
largest input integer
\cite{eisenbrand2019proximity,dense-subset-sum,icalp21,bounded-subset-sum}.

However, these parameterized results do not take advantage of structural
properties of the
input when the input numbers are very large. Therefore, we complement the
parameterization based on the input size by considering the parameterized complexity of Subset Sum with
respect to the doubling constant. This choice is natural not only because the
doubling constant is essential to the study of integer sets under addition, but
also because existing results from additive combinatorics give strong bounds on
the search space: Freiman's Theorem bounds the number of distinct subset sums of
an $n$-element input set by $n^{f(\DC)}$, where $f$ is a function that
depends only on $\DC$. 

The parameterization of Subset Sum in the cardinality of the solution $k$, otherwise known as $k$-SUM, has an entire literature of its own. Simple ``meet-in-the-middle'' algorithms that run in time $O(n^{\lceil k/2 \rceil})$ are conjectured to be optimal up to polynomial factors. The results of Abboud, Bringmann, and Fischer, and of Jin and Xu, suggest that the hardest instances of $k$-SUM are those with very little additive structure, such as Sidon sets \cite{abboud2022stronger,jin2022removing}. Parameterizing $k$-SUM in the doubling constant allows us to make analogous conclusions for the more general case of $k$-SUM: we can now prove results of the form, ``$k$-SUM instances with strong additive structure (i.e., small doubling constant) are easy''. 


\noindent
\subparagraph*{ Algorithms and Additive Combinatorics.}
This paper is also motivated by an emerging trend in fine-grained complexity and algorithms: ``importing'' results from additive combinatorics. In several recent works, researchers have achieved breakthroughs by taking existential results from the field of additive combinatorics and modifying their proofs to make them explicitly and efficiently constructive. 

For example, in 2015 Chan and Lewenstein proved a version of the Balog-Szemeredi-Gowers (BSG) theorem that allows certain sets guaranteed by the theorem to be constructed algorithmically \cite{chan2015clustered}. They then leveraged this result to solve the $(\min, +)$-convolution and 3-SUM problems on monotone sets of small integers. Recently, the  constructive BSG theorem found new applications. In 2022, Abboud, Bringmann and Fischer used this result, as well as a constructive version of Ruzsa's covering lemma, as a key ingredient in their proofs of lower bounds for approximate distance oracles and listing 4-cycles \cite{abboud2022stronger}. In the same year, Jin and Xu independently proved similar lower bounds and used the constructive BSG theorem to reduce 3-SUM to 3-SUM on Sidon sets \cite{jin2022removing}. More broadly, these works reflect the increasing role of additive combinatorics in algorithms over the last few decades; for general references, see \cite{Trevisan09,Viola11ac,Bibak13,Lovett17actcs}.

\subsection{Our Results}

\subparagraph{Contribution 1: \emph{A Constructive Freiman's Theorem in Near-Linear FPT Time.}} We begin
by unlocking a new tool to help us manipulate sets with significant
additive structure. \emph{Freiman's Theorem}, a cornerstone result in additive
combinatorics, states that every integer set with constant doubling is contained
inside a small (generalized) arithmetic progression. Naively constructing this generalized arithmetic progression takes \XP-time. We make the construction efficient by showing how an algorithm can obtain such an
arithmetic progression in time $\tilde{O}_\DC(n)$\footnote{ We use $O_\DC$ notation to
indicate the suppression of terms that depend only on $\DC$. For example,
$O_\DC(n^2) = f(\DC) \cdot O(n^2)$ for some computable function $f$. $\tilde{O}$ hides factors polylogarithmic in the argument, in this case $\log(n)$.} (\Cref{thm:Freiman-constructive}).
Later in the paper, we use this theorem to reduce Subset Sum with constant doubling to a constrained integer programming problem (Contribution 3) and to design efficient algorithms for Unbounded Subset Sum (Contribution 4). We hope that, like the
constructive BSG theorem discussed above, the constructive statement of
Freiman's Theorem may find other independent applications. 

\subparagraph*{Contribution 2: \emph{Integer Programming with Constant Doubling.}} An integer program specified by a constraint matrix $A \in \Z^{m \times n}$ and solution vector $b \in \Z^m$ is \emph{feasible} if there exists a solution $x \in \Z^n_{\geq 0}$ such that $Ax = b$. The ILP is \emph{binary} if the variables are further restricted to $x \in \{0,1\}^n$.

In our setting, we consider integer programs in which the set of column vectors
\[
    \mathcal{A} \coloneqq \{A[\cdot, j] \; \mid \; j \in [n]\}
\]
has constant doubling: $|\mathcal{A} + \mathcal{A}| \leq \DC |\mathcal{A}|$, for a constant $\DC$. We prove the following:
\begin{restatable}{theorem}{thmDCBILP}
    \label{thm:DC-BILP}
    An instance $\mathcal{I}$ of $\DC$-Binary ILP Feasibility on $n$ variables can be solved in time $n^{O_\DC(1)} \cdot \poly(|\mathcal{I}|)$.\footnote{ We write $|\mathcal{I}|$ to denote the size of the ILP instance $\mathcal{I}$. In the word RAM model (see \Cref{sec:prelims}), this is $\poly(m, n)$. }
\end{restatable}
This follows from Freiman's Theorem (without construction) and a dynamic
programming algorithm.
The theorem also holds when the variables $x_1, x_2, \dots, x_n$ have upper and lower bounds of magnitude $\poly(n)$. 

\subparagraph*{Contribution 3: \emph{Subset Sum with Constant Doubling.} } 
Our result for integer programming with constant doubling implies an
$n^{O_\DC(1)}$-algorithm for Subset Sum (\Cref{cor:dc-ss-xp}). 

Assuming the Exponential Time Hypothesis (ETH), there is no $2^{o(n)}$ time
algorithm for Subset Sum. Because $\DC = O(n)$, this means that we cannot hope
for a $2^{o(\DC)} n^{o(\DC /\log(\DC))}$ algorithm for $\DC$-Subset Sum under
the ETH. However, this lower bound does not exclude an $2^{O(\DC)} \cdot
n^{O(1)}$ algorithm.  A natural question is thus whether our upper bound can be
improved to an Fixed-Parameter Tractable (\textsf{FPT}) result: can $\DC$-Subset
Sum be solved in time $O_\DC(\poly(n))$? We show that this result appears
unlikely by way of an interesting connection to the feasibility of integer
programs with binary variables.


\begin{restatable}{theorem}{GBILequivalence}
    \label{thm:HBILP-SS-equivalence}
    $\DC$-Subset Sum can be solved in time $O_\DC(\poly(n))$ if and only if \emph{Hyperplane-Constrained Binary ILP} (HBILP) can be solved in time $\Delta^{O(m)} \cdot
    O_m(\poly(|\mathcal{I}|))$, where $|\mathcal{I}|$ is the size of the instance.
\end{restatable}

HBILP considers a constraint matrix $A
\in \Z^{m \times n}$ with entries bounded by $\Delta \coloneqq \norm{A}_\infty$, and asks whether there exists a solution $x \in \{0, 1\}^n$ such that $\langle Ax, s \rangle = t$ for a certain target $t$ and ``step vector'' $s$ orthogonal to a hyperplane. 
The best existing algorithm solves HBILP Feasibility in time $O_m(|\mathcal{I}|) + \Delta^{O(m^2)}$ (\cite{dadush2023optimizing}, Corollary 1\footnote{Corollary 2 in the arXiv preprint, 2303.02474.}). 

We can also reduce ILP Feasibility with bounded variables to HBILP feasibility (\Cref{lem:BILP-to-HBILP}). Thus \Cref{thm:HBILP-SS-equivalence} implies that an \FPT~algorithm for Subset Sum with constant doubling would imply a $\Delta^{O(m)} \cdot \poly(n)$ algorithm for ILP Feasibility with bounded variables (\Cref{cor:bilp-lb}). As previously noted in \cite{dadush2023optimizing}, reducing the exponent of $\Delta$ from $O(m^2)$ to $O(m)$ would be analogous to the recent improvement achieved by Eisenbrand and Weismantel for integer programs with \emph{unbounded} variables \cite{eisenbrand2019proximity}. 

Such an algorithm for ILP Feasibility would resolve the feasibility portion of one of the most significant open questions in the parameterized complexity of
integer programming: whether the $(\Delta^{O(m)} \cdot
O_m(|\mathcal{I}|))$-time algorithm for ILPs with unbounded variables can be
extended to ILPs with bounded variables
\cite{eisenbrand2019proximity,jansen2018integer,knop2020tight}. This would be a significant breakthrough in the area~\cite{jansen2018integer,knop2020tight}; accordingly, finding an FPT algorithm for $\DC$-Subset Sum is at least as difficult.


.

\subparagraph*{Contribution 4: \emph{Unbounded Subset Sum with Constant Doubling.} }
We can reduce an instance of Unbounded
Subset Sum with constant doubling to an ILP with $m$ constraints, $n$ binary
variables, and entries of $A$ bounded by $\Delta = n^{O(1/d(\DC))}$ using our constructive
Freiman's theorem. Because
solvable ILPs with bounded $\Delta$ admit solutions with small support, this
allows us to solve Unbounded Subset Sum in time $n^{O_\DC(\log\log\log n)}$, or
$n^{O_\DC(1)}$ under the hypothesis that a $v$-variable ILP $\mathcal{I}$ can be
solved in time $2^{O(v)}\poly(\mathcal{I})$ (\Cref{thm:dc-uss-xp}).

\subparagraph*{Contribution 5: \emph{$k$-SUM with Constant Doubling.}} The
application of recent algorithms for sparse nonnegative convolution
\cite{bringmann2022deterministic} allow us to efficiently solve $k$-SUM with
constant doubling in time {$\tilde{O}(\DC^{\lceil k/2 \rceil} \cdot 2^{O(k)} \cdot n)$}
(see~\Cref{thm:ksum}).

Because the $k$-SUM conjecture implies a lower bound of $\Omega(\DC^{\lceil k/2
\rceil -1} n)$, this leaves a $\DC$-factor gap. Part of the gap can be explained
by the fact that the Pl\"unnecke-Ruzsa inequality, which we use to derive the
upper bound, does not give the optimal exponent for $\DC$; applying recent improvements to the inequality
narrows the gap slightly. In the specific case of $(\DC, 4)$-SUM, our algorithm
achieves a runtime of $\tilde{O}(\DC n)$, which is optimal up to polylogarithmic
factors under the $k$-SUM conjecture.

\subsection{Organization }

We begin with mathematical preliminaries in \Cref{sec:prelims}, although some definitions required for the constructive proof of Freiman's Theorem in \Cref{sec:construct-Freiman} are deferred to the proof of this result in \Cref{apx:freiman-constructive}. In \Cref{sec:ILP}, we present our algorithms for ILP feasibility with bounded doubling. Finally, we present our bounds for Subset Sum in \Cref{sec:dc-subset-sum}, Unbounded Subset Sum in \Cref{sec:dc-uss}, and $k$-SUM in \Cref{sec:ksum}.

\section{ Preliminaries }
\label{sec:prelims}

\textbf{RAM Model.} Throughout the paper, we use the standard \emph{word RAM} model, in which input integers fit into a single machine word and logical and arithmetic operations on machine words take time $O(1)$. If we make the weaker assumption that operations on $b$-bit words take $\polylog(b)$ time, this adds a $\polylog(b)$ factor to \Cref{thm:Freiman-constructive} and the results that rely on it. 

\vspace{0.3cm}
\noindent
\textbf{Big-$O$ Notation.} We use $O_\DC$ notation to indicate we have
suppressed terms that depend only on $\DC$. For example, $O_\DC(n^2) = f(\DC) \cdot O(n^2)$ for some computable function $f$. $\tilde{O}$ notation suppresses polylogarithmic factors of $n$ and $\Delta$: for instance, $n \log^2(n) = \tilde{O}(n)$.

\vspace{0.3cm}
\noindent
\textbf{Sets.} We write $[n]$ for the integer set $\{1, 2, \dots, n\}$ and
$[a:b]$ (with $a \leq b$) for the integer set $[a, a+1, a+2, \dots, b]$. The
\emph{diameter} of an integer set $A$, denoted $\mathrm{diam}(A)$, is $\max_{a,
b \in A} |a - b|$. We write $\Sigma(X)$ as shorthand for the sum of elements $\sum_{x \in X} x$, and
$\Sigma(2^X)$ as shorthand for the set of subset sums $\{ \Sigma(X') \; : \; X' \subseteq X\}$. 

\vspace{0.3cm}
\noindent
\textbf{Vectors.}
Given a vector $x \in \Z^n$, we define $\supp(x) \subseteq [n]$ to be the set
of non-zero coordinates of $x$.

For $a,b \in \ZZ^n$ we say that $a$ is \emph{lexicographically prior} to $b$, denoted $a \lelex b$, if and only if there exists $k \in [n]$ such that $a[k] < b[k]$ and for every $1\le i < k$ it holds that $a_i = b_i$. Observe that $\lelex$ is a total order and that every set of vectors $S \subseteq \ZZ^n$ contains a unique element that is lexicographically minimal.

\vspace{0.3cm}
\noindent
\textbf{Matrices.} Given a $m \times n$ matrix $A$, we write $A[i, j]$ to denote the component of $A$ at row $i$, column $j$. We write $A[i, \cdot]$ and $A[\cdot, j]$ to denote the $i$th row and $j$th column of $A$, respectively.

We write $J_{m \times n}$ to denote the $m \times n$ matrix in which each entry is 1.

\vspace{0.3cm}
\noindent
\textbf{Group Theory and Linear Algebra.} Given an integer $m$, we write $\Z_m$ to denote the cyclic group of order $m$ (under addition). When $p$ is prime, every element of $\Z_p$ is a generator except for 0.

A \emph{lattice} in $\R^d$ is defined by $d$ linearly independent vectors $v_1, v_2, \dots, v_d \in \R^d$, collectively referred to as the \emph{basis} of the lattice. The lattice itself is the set 
\[
    \Lambda = \left\{ \sum_{i \in [d]} a_i v_i \; \middle| \; a_i \in \Z \right\}
\]
of all integer linear combinations of $v_1, v_2, \dots, v_d$, and each point in $\Lambda$ is referred to as a \emph{lattice vector}. 

The \emph{determinant} of a lattice, denoted $\det(\Lambda)$, is the determinant of the matrix whose columns are the lattice basis. Geometrically, $\det(\Lambda)$ is the volume of the \emph{fundamental parallelepiped} spanned by the lattice basis. In general, if $T$ is a convex body, we write $vol(T)$ to denote the volume of $T$.

Given two $m$-dimensional vectors $x$ and $y$, we use $\langle x, y \rangle$ to denote the dot product $x_1y_1 + \dots + x_my_m$.

\vspace{0.3cm}
\noindent
\textbf{Norms.} Given a real number $r$, we write $\norm{r}_{\R / \Z}$ to denote
the distance from the nearest integer. Given a finite-dimensional vector $v$,
the $L_\infty$ norm $\norm{v}_{\infty}$ denotes the largest absolute value of any coordinate.

\vspace{0.3cm}
\noindent
\textbf{Additive Combinatorics.} 
Given an integer set $X$, $X + X$ denotes the \emph{sumset} $\{a + b \; : \; a, b \in X\}$. We write $sX$, where $s$ is a positive integer, as shorthand for the iterated sumset $\underbrace{X + \ldots + X}_{s \text{ times}}$.

A \emph{generalized arithmetic progression} (GAP) $P$ is an integer set
\[
    P = \{\ell_1 y_1  + \ell_2 y_2 + \dots + \ell_d y_d \; : \; 0 \leq \ell_i < L_i, \forall i \in [d]\},
\]
defined by the integer vector $y = \{y_1, y_2, \dots, y_d\}$ and the dimension bounds $L_1, L_2, \dots, L_d$. We say that $P$ has \emph{dimension} $d$ and \emph{volume} $\prod_{i \in [d]} L_i$. When we write that an algorithm ``explicitly constructs'' or ``returns'' $P$, we mean specifically that the algorithm computes $y_i$ and $L_i$ for all $i \in [d]$. 

We can think of $P$ as a projection of a $d$-dimensional parallelepiped onto the line. $P$ is \emph{proper} if $|P| = L_1 L_2\dots L_d$, that is, if each point in the parallelepiped projects to a unique point on the line. 


\section{Freiman's Theorem Made Constructive in FPT Time}
\label{sec:construct-Freiman}

Freiman's Theorem states that any integer set $X$ with constant doubling is contained inside a generalized arithmetic progression of constant dimension and volume at most $|X|$ times a constant.

\begin{theorem}[Freiman's Theorem, \cite{freiman1964addition}, see  \cite{zhao2022graph} for a modern presentation]
    \label{thm:freiman}
    Any finite integer set $X$ with $|X + X| \leq \DC|X|$ is contained in a
    GAP $P$ of dimension $d(\DC)$ and volume $v(\DC)|X|$, where $d$ and $v$ are computable functions that depend
    only on $\DC$.
\end{theorem}

We make this statement constructive by showing an algorithm that, given $X$, can explicitly construct the progression $P$ in \FPT~time. In fact, the construction is near-linear, losing only a $\polylog(n)$ factor and a (large) function of $\DC$. 

\begin{theorem}[\textsf{FPT} Freiman's Theorem]
    \label{thm:Freiman-constructive}
    Let $A$ be a set of $n$ integers satisfying $|A+A| \leq \DC |A|$. There exists an $\tilde{O}_\DC(n)$ algorithm that, with probability $1 - n^{-\gamma}$ for an arbitrarily large constant $\gamma > 0$, returns\footnote{ Specifically, we compute the values $x_1, x_2, \dots, x_{d(\DC)}$ and $L_1, L_2, \dots L_{d(\DC)}$. } an arithmetic progression
         \[
             P = \{x_1\ell_1 + x_2\ell_2 + \dots + x_{d(\DC)}\ell_{d(\DC)} \; : \; \forall i, \ell_i \in [L_i] \} \supseteq A
         \]
    with dimension $d(\DC)$ and volume $v(\DC) \cdot |A|$, where $d$ and $v$ are computable functions that depend
    only on $\DC$.

    (We make the standard assumption that arithmetic operations on integers require $O(1)$ time.)
\end{theorem}

In outline, the proof proceeds as follows:
\begin{itemize}[nosep]
    \item \emph{Step 1}: We prove a constructive version of \emph{Ruzsa's Modeling Lemma}, which allows us to map our integer set to a small cyclic group such that additive structure is preserved.
    \item \emph{Step 2}: We prove a constructive version of \emph{Bogolyubov's Lemma}, which allows us to find a Bohr set contained in the cyclic group. Roughly speaking, the Bohr set (1) behaves ``like a subspace'' and (2) is within a constant factor of the size of our original set. This step requires the Fast Fourier Transform (FFT).
    \item \emph{Step 3}: Finding a large GAP within our Bohr set requires finding a small basis for a certain lattice. Fortunately, the lattice has dimension $O_\DC(1)$, so we can enumerate the entire set of short lattice vectors.
    \item \emph{Step 4}: At this point we are left with a GAP that covers a constant fraction of the image of our original set in the cyclic group. Using \emph{Ruzsa's Covering Lemma}, previously made efficiently constructive by Abboud, Bringmann, and Fischer \cite{abboud2022stronger}, we can find a GAP that covers the entire image of our input set. We then map back to the integers to complete the construction.
\end{itemize}

We defer the full proof to \Cref{apx:freiman-constructive}. The following observation further simplifies \Cref{thm:Freiman-constructive}.

\begin{observation}
    In the GAP $P$ guaranteed by \Cref{thm:Freiman-constructive}, without loss
    of generality we can assume 
    \[
        L_i \leq n^{2/d(\DC)}
    \]
    for all $i \in [d(\DC)]$, where $d(\DC)$ denotes the dimension of $P$.
    \label{obs:bigL-bound}
\end{observation}

We defer the proof of~\Cref{obs:bigL-bound} to \Cref{apxsec:bigL-bound}.

\section{Integer Programming with Constant Doubling}
\label{sec:ILP}

For an integer program, we consider the doubling constant of the \emph{column set} of the constraint matrix $A$ as our parameter. This is because the column is the smallest unit affected by each variable $x_i$ when we compute the product $Ax$; as a result, duplicate columns in $A$ play a similar role to duplicate elements in a Subset Sum instance, and indeed can often be eliminated without loss of generality. This formulation allows $A$ to contain duplicate entries (for example, multiple $0$'s and $1$'s) as long as all columns are distinct.

Given a matrix $A$, we use the shorthand
\[
    \mathcal{A} \coloneqq \mathcal{A}(A) = \{ A[\cdot, j] \; \mid \; j \in [n] \}
\]
to denote the set of column vectors of $A$. Vector set addition (that is,
$\mathcal{A} + \mathcal{A}$) is defined in the natural way, using vector instead
of integer addition. 

\begin{prob}[$\DC$-Integer Linear Programming (ILP) Feasibility]{prob:ILP}
    \textbf{In:} An integer linear program specified by an integer matrix $A \in
    \Z^{m \times n}$ with $n$ distinct columns and an integer target $b
    \in \Z^m$, such that the column set $\mathcal{A} \coloneqq \mathcal{A}(A)$ satisfies 
    $
        |\mathcal{A} + \mathcal{A}| \leq \DC |\mathcal{A}|
    $
    for a constant $\DC$ independent of $m$ and $n$. \\
    \textbf{Out:} Vector $x \in \ZZ^n$ such that $Ax = b$, or `NO' if no solution exists.
\end{prob}

If each variable $x_i$ is constrained to satisfy $x_i \in [\ell_i : u_i]$, where $\ell_i$ and $u_i$ indicate the lower and upper bounds of a range of valid variable assignments, we refer to the problem as \textbf{$\DC$-Bounded ILP Feasibility}. Further restricting the variables to $x \in \{0,1\}^n$ yields \textbf{$\DC$-Binary ILP Feasibility}.

\begin{remark}
    \label{remark:bounded-to-binary}
Bounded ILPs with $n$ variables and $|\ell_i|, |u_i| = O(\poly(n))$ for $i \in [n]$ can be converted into equivalent binary ILPs with $\poly(n)$ variables by duplicating the columns of $A$.
\end{remark}

\subsection{ \texorpdfstring{$\DC$}{}-Binary ILP Feasibility }

Given a constraint matrix with constant doubling, Freiman's Theorem bounds the
number of possible values for $Ax$ corresponding to any variable assignment if
the variables are binary or bounded. This allows us to solve the problem
efficiently via dynamic programming, and does not actually require constructing
the GAP guaranteed by Freiman's Theorem.\footnote{The constructive Freiman's theorem will be required later, specifically in ~\Cref{lem:SS-to-HBILP,thm:dc-uss-xp}. The current result emphasizes the usefulness of parameterization
in the doubling constant.}

\thmDCBILP*
\begin{proof}
Fix an instance of $\DC$-Binary ILP feasibility specified by $A \in \Z^{m \times n}$ and $b \in \Z^m$, with the column set $\mathcal{A}$ satisfying $|\mathcal{A} + \mathcal{A}| \leq \DC |\mathcal{A}|$.

Let $L \coloneqq \Sigma(2^{\mathcal{A}})$ denote the list of all (vector) sums that can
be attained by adding together any subset of the columns of $A$. Equivalently, this is the set of possible outputs $Ax$ for any $x \in \{0,1\}^n$. Our first goal is to bound $|L|$.

First, we observe that there exists a GAP $P$ with dimension $d(\DC)$ and volume $v(\DC) n$ such that 
\[
    \mathcal{A} \subseteq P = \{x_1k_1 + x_2k_2 + \dots + x_{d(\DC)}k_{d(\DC)} \; : \; \forall i, k_i \in [K_i] \},
\]
where $x_i \in \Z^m$ for all $i \in [d(\DC)]$. This is true even though $\mathcal{A}$ is a set of integer vectors, as Freiman's Theorem holds for torsion-free\footnote{That is, groups in which only the identity element has finite order.} commutative groups (\cite{ruzsa2009sumsets}, Theorem 8.1).

Thus $L$ is contained in the GAP
\[
    P' = \{x_1k_1 + x_2k_2 + \dots + x_{d(\DC)}k_{d(\DC)} \; : \; \forall i, k_i \in [n \cdot K_i] \},
\]
which implies
\begin{equation}
    \label{eq:LA-bound}
    |L| \leq |P'| \leq n^{d(\DC)} |P| = n^{d(\DC)} v(\DC) n = n^{O_\DC(1)}.
\end{equation}

To complete the proof, we claim that we can enumerate $L$ efficiently via dynamic programming, using the following procedure:
Initially, we set $L_1 = A[1,\cdot]$. Then, we iterate $i = 2,3,\ldots,n$. In the
$i$th iteration, we construct the sorted list $L_i$, defined as:

\begin{displaymath}
    L_i \coloneqq L_{i-1} \cup \{a+A[i,\cdot] \mid a \in L_{i-1}\}.
\end{displaymath}

Finally, we return list $L = L_n$. Correctness of the above algorithm follows
immediately by a construction. For the running time, observe that $L_i$ can be
constructed in $O(|L_i|)$ time.
Because each of the $n$ iterations of the subprocedure takes time $O(|L_i|) = O(|L|)$, the total runtime is, by \eqref{eq:LA-bound}, at most
$n \cdot O(|L|) = n^{O_\DC(1)}$.
\end{proof}

\subsection{\texorpdfstring{$\DC$}{}-Bounded ILP Feasibility}
In general, ILPs with polynomially bounded variables can be converted to ILPs with binary
variables (see \Cref{remark:bounded-to-binary}); however, the straightforward
reduction can create many duplicate columns in the resulting Binary ILP. Although it is possible to get rid of the duplicate columns, it is easier to extend the previous result to $\DC$-Bounded ILP Feasibility directly:

\begin{corollary}
    \label{cor:DC-BoundedILP}
    An instance $\mathcal{I}$ of $\DC$-Bounded ILP Feasibility such that $\ell_i \le
    x_i \le u_i$ and $|\ell_i|, |u_i| = \poly(n)$ for $i \in [n]$ can be solved in time $n^{O_\DC(1)} \cdot \poly(|\mathcal{I}|)$.
\end{corollary}
\begin{proof}
    Modify the proof of \Cref{thm:DC-BILP} by considering the list $L'$ of all possible outputs $Ax$ for each valid assignment of variables $x$, using the variable bounds $x_i \in [\ell_i : u_i]$ for $i \in [n]$ instead of $x \in \{0,1\}^n$. As before, we bound $|L'|$.
    
    Observe that $L'$ is contained in the GAP $P''$ obtained by scaling each range bound $L_i$ of $P$ by a factor of $n^{O(1)}$, where the hidden constant is determined by the bounds on the variables. It follows that $|L'| = n^{O_\DC(1)}$. We can enumerate $L'$ by modifying the procedure given above so that Step 2 merges a polynomial number of lists, one for each variable assignment.
\end{proof}


\section{Subset Sum with Constant Doubling }
\label{sec:dc-subset-sum}

We now consider the useful applications of parameterization in the doubling constant to Subset Sum. Formally, we consider the following problem:

\begin{prob}[$\DC$-Subset Sum]{prob:KSS}
    \textbf{In:} An integer set $Z = \{z_1, z_2, \dots, z_n\}$ such that $|Z+Z| \leq \DC|Z|$ and an integer target $t$. \\
    \textbf{Out:} $S \subseteq Z$ such that $\Sigma(S) = t$, or `NO' if no solution exists.
\end{prob}

$\DC$-Subset Sum is equivalent to $\DC$-Binary ILP with a single constraint. As a result, \Cref{thm:DC-BILP} yields the following corollary for Subset Sum with $n$ variables:

\begin{corollary}[$\DC$-Subset Sum is in XP]
    \label{cor:dc-ss-xp}
    $\DC$-Subset Sum can be solved in time $n^{O_\DC(1)}$.    
\end{corollary}

At this point, it is natural to wonder whether $\DC$-Subset Sum can be
solved in time $O_\DC(1) \cdot n^{O(1)}$: that is, whether Subset Sum is in
\FPT with respect to the doubling constant. While we cannot yet prove or
disprove this statement, we can show that it is equivalent to an open
problem in the parameterized complexity of integer programming. The
remainder of this section proves this reduction in both directions.

\subsection{ Reduction from \texorpdfstring{$\DC$}{}-Subset Sum to Hyperplane-Constrained Binary ILP Feasibility }
\label{subsec:SS-to-HBILP}

    Recent generalizations of Integer Programming consider the problem of optimizing the value $g(Ax)$ in place of $Ax$, where $g: \R^m \rightarrow \R$ is a low-dimensional objective function \cite{dadush2023optimizing}. The mapping given by Freiman's Theorem provides a natural reduction from Subset Sum with constant doubling to a problem of this form. Specifically, $\DC$-Subset Sum reduces to a Binary ILP feasibility problem in which the constraint matrix $A$ has bounded entries and a feasible solution is any $x$ satisfying $\langle Ax, s \rangle = t$ for a specific ``step vector'' $s$. Formally, our problem is as follows:

    \begin{prob}[Hyperplane-Constrained Binary ILP (HBILP) Feasibility]{prob:planar-ILP}
        \textbf{In:} An integer matrix $A \in \Z^{m \times n}$, a step vector $s \in \Z^m$, and a target integer $t $. We let $\Delta \coloneqq \norm{A}_\infty$, the magnitude of $A$'s largest entry. \\
        \textbf{Out:} A vector $x \in \{0,1\}^n$ such that $\langle Ax, s \rangle = t$, or `NO' if no solution exists.
    \end{prob}

    The reduction from $\DC$-Subset Sum to HBILP Feasibility (\Cref{lem:SS-to-HBILP}) is straightforward but relies crucially on our constructive Freiman's Theorem.

    \begin{lemma}
        \label{lem:SS-to-HBILP-reduction}
        For any fixed instance $(Z, t)$ of $\DC$-Subset Sum, there exists a HBILP Feasibility instance given by $A \in \Z^{d(\DC) \times n}$, $s \in \Z^{d(\DC)}$, and $t$ for some function $d(\DC)$ such that a vector $x \in \{0,1\}^n$ satisfies 
        \[
            Ax = t \quad \text{ if and only if } \quad \sum_{i \; : \; x_i = 1} z_i = t.
        \]
        Moreover, $\Delta \coloneqq \norm{A}_\infty \leq n^{2 / d(\DC)}$, and the reduction can be computed in time $\tilde{O}_\DC(n)$ with success probability $1 - n^{-\gamma}$ for an arbitrarily small constant $\gamma$.
    \end{lemma}
    \begin{proof}
    Fix an instance of $\DC$-Subset Sum given by an integer set $Z$ satisfying
    $|Z + Z| \leq \DC |Z|$ and an integer target $t$. Apply
    \Cref{thm:Freiman-constructive}, which fails with probability $n^{-\gamma}$
    and otherwise produces a GAP
    \[
        P = \{y_1\ell_1 + y_2\ell_2 + \cdots + y_d\ell_d \; : \; \forall i, \ell_i \in [L_i]\}
    \]
    of dimension $d \coloneqq d(\DC)$ and volume $v(\DC)n$ containing $Z$. 

    For each $z_i \in Z$, let $v(z_i) = (v_1, v_2, \dots, v_d)$
    be an arbitrary $d(\DC)$-dimensional integer vector satisfying
    \begin{align*}
        y_1v_1 + y_2v_2 + \dots y_dv_d &= z_i \text{ and } \\
        \forall i \in [d], v_i &\in [L_i].
    \end{align*}
    We can think of $v(z_i)$ as the $d$-dimensional ``GAP coordinates'' of the input element $z_i$. $v(z_i)$ is guaranteed to exist by Freiman's theorem, and we can recover it in time $O(|P|) = O_\DC(n)$ via exhaustive search of $P$. (However, $v(z_i)$ is not guaranteed to be unique.)
    
    To complete the reduction, set 
    \begin{align*}
        A &\in \Z^{d(\DC) \times n} \text{ with } \forall j \in [n], A[\cdot, j] = v(z_j), 
    \end{align*}
    set $s \coloneqq (y_1, y_2, \dots, y_d)$ and preserve the same target $t$. Note that $\norm{A}_\infty \leq n^{2 / d(\DC)}$ without loss of generality by \Cref{obs:bigL-bound}.

    We claim that for any binary vector $x = (x_1, x_2, \dots, x_n) \in \{0, 1\}^n$,
    \begin{equation}
        \label{eq:equivalence-SS-HBILP}
        \langle Ax, s \rangle = \sum_{i \; : \; x_i = 1} z_i.
    \end{equation}
    To see this, observe that
    \begin{align*}
        \langle Ax, s \rangle &= \sum_{i \in [d]} y_i \sum_{x_j = 1 } A[i, j] \\
        &= \sum_{x_j = 1} y_1A[1, j] + y_2 A[2, j] + \dots y_d A[d,j] \\
        &= \sum_{x_j = 1} \langle y, v(z_j) \rangle \\
        &= \sum_{j \; : \; x_j = 1} z_j.
    \end{align*}
    Thus $\langle Ax, s \rangle = t$ if and only if $\sum_{i \; : \; x_i = 1} z_i = t$, and there is a one-to-one correspondence between solutions to our $\DC$-Subset Sum instance and our HBILP feasibility instance.
    \end{proof}

\subsection{Equivalence Between HBILP Feasibility and Subset Sum}
\label{subsec:HBILP-to-SS}

\GBILequivalence*

\Cref{thm:HBILP-SS-equivalence} follows immediately from the next two lemmas, which show reductions in both directions. The first is a consequence of the reduction in \Cref{subsec:SS-to-HBILP}:

\begin{lemma}
    If HBILP Feasibility can be solved in time $\Delta^{O(m)} \cdot
    O_m(\poly(n))$, then $\DC$-Subset Sum can be solved in time $O_\DC(\poly(n))$ with success probability $1 - n^{-\gamma}$ for an arbitrarily large constant $\gamma > 0$.
    \label{lem:SS-to-HBILP}
\end{lemma}
\begin{proof}
    In polynomial time (in the size of the input), we can preprocess an instance
    of Subset Sum and produce an equivalent one such that all integers are bounded by $2^{\poly(n)}$ (see, e.g.,
    ~\cite{frank1987application,doi:10.1137/060668092}).\footnote{See discussion
    about the computational model in~\Cref{sec:prelims}.}
    Next, we use the reduction given in \Cref{lem:SS-to-HBILP-reduction}, which takes time $\tilde{O}_\DC(n)$ and succeeds with probability $1 - n^{-\gamma}$, and solve the resulting HBILP instance in time 
    \[
        \Delta^{O(m)} \cdot O_m(\poly(n)) = (n^{2/d(\DC)})^{O(d(\DC))} \cdot
        O_\DC(\poly(n)) = O_\DC(\poly(n)). \qedhere
    \]
\end{proof}

\begin{lemma}
    If $\DC$-Subset Sum admits an $O_\DC(\poly(n))$-time algorithm, HBILP
    Feasibility can be solved in time $\Delta^{O(m)} \cdot O_m(\poly(n))$.
    \label{lem:HBILP-to-SS}
\end{lemma}
\begin{proof}
    Fix an instance of HBILP Feasibility given by the matrix $A \in \Z^{m \times
    n}$, the vector $s \in \Z^m$, and the integer target $t$. Let $\Delta
    \coloneqq \norm{A}_\infty$. 

    We perform the reduction in two steps. First, we self-reduce our HBILP instance to another HBILP instance $A', s', t'$ with the property that every column $A'[\cdot, j]$ of $A'$ has a unique dot product $\langle A'[\cdot, j], s' \rangle$. We then reduce $A', s', t'$ to $\DC$-Subset Sum.

    If $A$ contains any column with only zeroes, then the value of the corresponding entry of $x$ does not matter, and we can safely delete it. Thus we can assume without loss of generality that each column of $A$ has at least one nonzero entry. Moreover, by~\Cref{obs:HBILP-non-negative}, proved in \Cref{apx:ilp-manipulation}, we can assume that each entry of $A$ is non-negative and that any solution vector $x$ has fixed support exactly $q$ for some $q = \Theta(n)$.
    
    \vspace{0.3cm}
    \noindent
    \textbf{Step 1: Self-reduction.} In order to construct the instance $A', s',
    t'$, define $M \coloneqq nm\Delta \norm{s}_\infty + 1$, which satisfies 
    \begin{equation}
        \label{eq:M-lb}
        M > \langle Ax, s \rangle
    \end{equation} 
    for any $x \in \{0, 1\}^n$ by construction. Moreover, let $k \coloneqq \lceil \log_\Delta(n) \rceil$.
    
    Let $R \in [\Delta]^{k \times n}$ be the matrix whose columns are vectors in $[0:\Delta-1]^k$ in lexicographically increasing order. Because the number of such
    vectors is at least $n$, every column of $R$ is different.
    Recall that $J_{k\times n}$ denotes the $k \times n$ matrix containing only 1's and let $\bar{R}$ be $\Delta \cdot J_{k\times n}- R$. 

    Create the block matrix $A' \in \Z_{\ge 0}^{ (m+k) \times 2n}$ as follows. The top-left block is $A$, the bottom-left block is $R$, the bottom-right block is $\bar{R}$ and each entry in the top-right block is $0$. Observe that every column in $\tilde{A}$ is distinct because each column in $R$ is distinct and no column in $A$ is all 0's. 

    Create $s' \in \Z_{\ge 0}^{m + k}$ as follows. The first $m$ entries
    of $s'$ are $s$, and the remaining $k$ entries are the vector $v =
    (M\Delta^0, M \Delta^1, \dots, M \Delta^{k-1})$. Finally, set $t' \coloneqq t + q
    \Delta \norm{v}_1$ to complete the reduction.
    \[
	A' \coloneqq 
	\begin{pmatrix}
		\makebox(1.5cm,1.5cm){ $A$}
		& \rvline & \makebox(1.5cm,1.5cm){\bigzero} \\
		\hline
		\makebox(1.5cm,0.5cm){$R$} & \rvline &
		\makebox(1.5cm,0.5cm){$\bar{R}$} 
	\end{pmatrix}
	\smallskip
	\,\,\,
	s' \coloneqq
	\begin{pmatrix}
		\makebox(1.5cm,1.5cm){ $s$ } \\
		\hline
		\makebox(1.7cm,0.5cm){$v$} \\
	\end{pmatrix}
    \]

    \begin{claim}
        \label{claim:Aprime-products-distinct}
        For every distinct pair of indices $i, j \in [2n]$, $\langle A'[\cdot, i], s' \rangle \neq \langle A'[\cdot, j], s' \rangle$.
    \end{claim}
    \begin{claimproof}
        Begin with the first $n$ columns. For all $i \in [n]$, we can break down the relevant dot product into two pieces corresponding to the top and bottom portions of $A'$: 
        \[
            \langle A'[\cdot, i], s' \rangle = \langle A[\cdot, i], s \rangle + \langle R[\cdot, i] \cdot v \rangle.
        \]
        
        First, observe that $\langle R[\cdot, i], v \rangle$ is distinct for every $i \in [n]$ by construction, due to the fact that each component of $R[\cdot, i]$ is less than $\Delta$, and the components of $v$ increase by factors of $\Delta$. 
        
        Second, because
        \[
            0 < \langle A[\cdot,  i], s\rangle \leq M,
        \]
        the $\langle A[\cdot,  i], s\rangle$ term of the dot product $\langle A'[\cdot,  i], s'\rangle$ is not large enough to interfere with the $\langle R[\cdot, i], v \rangle$ term, and thus the first $n$ columns of $A'$ have distinct dot products with $s'$.
        
        Because $\bar{R} = \Delta \cdot J_{k \times n} - R$, and because no column of $A$ consists of all $0$'s by assumption, similar arguments show that the value $\langle A'[\cdot, i], s' \rangle$ is distinct for every column $i \in [2n]$.
    \end{claimproof}
    
    \begin{claim}
        \label{claim:Aprime-selfreduction}
        The ILP instance $(A', s', t')$ has a solution if and only if the
        instance $(A, s, t)$ has a solution (and the solution to $A, s, t$ can
        be recovered efficiently from the solution of $A', s', t'$). 
    \end{claim}
    \begin{claimproof}
        Suppose $x$ satisfies $\langle Ax, s \rangle = t$. Recall that $x$ has support exactly $q$ by \Cref{obs:HBILP-non-negative} without loss of generality. Thus the vector $x'$ created by concatenating two copies of $x$ satisfies 
        \[
            \langle A'x', s' \rangle = t + q \Delta \norm{v}_1 = t'.
        \]
        Moreover, any vector $y' \in \{0,1\}^{2n}$ that satisfies $\langle A'y', s' \rangle = t'$ must satisfy
        \[
            \langle A(y'_1, y'_2, \dots, y'_n), s \rangle = t
        \]
        by construction. This is because the $[R \mid \bar{R}]$ submatrix of $A'$ can contribute to $t'$ only in multiples of $M$, so because $\langle A(y'_1, y'_2, \dots, y'_n), s \rangle < M$ by \eqref{eq:M-lb}, this product must evaluate to $t$.
    \end{claimproof}
    
    \vspace{0.3cm}
    \noindent
    \textbf{Step 2: Reduction to $\DC$-Subset Sum.}    
    Consider the integer vector 
    \begin{equation}
        \label{eq:z-def}
        z \coloneqq (\langle s', A'[\cdot, 1] \rangle, \langle s', A'[\cdot, 2] \rangle, \dots, \langle s', A'[\cdot, 2n] \rangle)
    \end{equation}
    and let 
    \[
        Z = \{ z_1, z_2, \dots, z_{2n} \}
    \]
    denote the set containing the components of $z$. (Note that $Z$ is a proper set and contains no duplicates, by \Cref{claim:Aprime-products-distinct}.) We proceed to consider $Z, t$ as an instance of Subset Sum.

    Because $\langle A'x, s' \rangle = t'$ if and only if $\langle x, z \rangle =
    t'$ by construction \eqref{eq:z-def}, we have a one-to-one correspondence between solutions
    to our Subset Sum and HBILP Feasibility instances: any subset of $Z$ that
    adds to $t'$  corresponds to a binary vector $x \in \{0, 1\}^{2n}$ such that
    $\langle A'x, s' \rangle = t'$, which can be used to recover a solution for the original instance $A$, $s$, $t$ by \Cref{claim:Aprime-selfreduction}. It remains to show that an $O_\DC(\poly(n))$
    algorithm for $\DC$-Subset Sum will allow us to solve the problem in the
    claimed time.
    
    We begin by bounding the doubling constant $\DC$ of $Z$. By the definition of $z$, we have that $z_j = \langle s, A'[\cdot, j] \rangle$ for all $j \in [n]$, and thus $Z$ is a subset of the GAP
    \[
        Y \coloneqq \{y_1s_1 + y_2s_2 + \dots + y_ms_m + y_{m+1}M\; : \; -\Delta \leq y_i \leq \Delta, \forall i \in [m]; 0 < y_{m+1} < \Delta^{k-1} \}.
    \]
    Note here that the dimension of $Y$ is $m+1$ instead of $m+k$, as we have chosen to represent the component of each $z_j \in Z$ divisible by $M$ into a single large dimension.

    We claim that we can assume 
    \begin{equation}
        \label{eq:Z-Omega-Y}
        |Z| = \Omega(|Y|)
    \end{equation}
    without loss of generality. To see this, observe that we can inflate $|Z|$ by adding up to $|Y|$ dummy elements from the translated GAP $t + Y$. Because every such element is greater than $t$, and each is contained in a translation of $Y$, we create no additional solutions and increase $|Y+Y|$ by at most a factor of 2.
    
    We have that 
    \begin{align*}
        |Z + Z| &\leq |Y + Y| \\
        &\leq 2^{m+1} \cdot |Y| \\
        &= O_m(1) \cdot |Z|,
    \end{align*}
    where the first line follows from the fact that $Z \subseteq Y$, the second line follows from the fact that $|Y|$ has dimension $m+1$, and the third follows from \eqref{eq:Z-Omega-Y}.

    Thus $Z, t$ is an instance of $O_m(1)$-Subset Sum whose solutions correspond directly to solutions of our original HBILP feasibility instance. Also, $|Z| = O(|Y|) = \Delta^{O(m) + k}$.
    Because $\Delta^k = O(n)$ by the definition of $k$, an algorithm for Subset Sum that runs in time $O_\DC(\poly(n))$ solves $Z, t$ in time $O_m(\poly(\Delta^m \cdot n)) = \Delta^{O(m)} \cdot O_m(\poly(n))$ as claimed.
\end{proof}

\subsubsection{Reduction from BILP Feasibility to HBILP Feasibility}
\label{subsubsec:BILP-to-HBILP}

An \FPT~algorithm for $\DC$-Subset Sum further implies a $\Delta^{O(m)} \cdot O_m(\poly(n))$ algorithm for Bounded ILP feasibility, i.e., an extension of Eisenbrand and Weismantel's improvement for Unbounded ILPs to determining feasibility for Bounded ILPs.

\begin{corollary}
    \label{cor:bilp-lb}
    If $\DC$-Subset Sum can be solved in $O_\DC(\poly(n))$, then Bounded ILPs
    defined by $A \in \Z^{m \times n}$, $b \in \Z^m$ with $\Delta \coloneqq \norm{A}_\infty$ and each variable $x_i$ bounded by $\poly(n)$ can be solved in time $\Delta^{O(m)} \cdot O_m(\poly(n))$.
\end{corollary}

\Cref{cor:bilp-lb} is a straightforward corollary of \Cref{lem:BILP-to-HBILP},
which reduces ILP Feasibility with binary variables to HBILP feasibility, and
the fact that ILPs with polynomially bounded variables can be reduced to binary ILPs (\Cref{remark:bounded-to-binary}). We defer the proof of \Cref{lem:BILP-to-HBILP} to \Cref{apx:ilp-manipulation}.

\begin{lemma}
    \label{lem:BILP-to-HBILP}
    If HBILP Feasibility can be solved in time $\Delta^{O(m)} \cdot O_m(\poly(n))$, Binary ILP Feasibility can be solved in time $\Delta^{O(m)} \cdot O_m(\poly(n))$. 
\end{lemma}

\section{Unbounded Subset Sum with Constant Doubling}
\label{sec:dc-uss}

$\DC$-Unbounded Subset Sum is equivalent to an unbounded integer program with a single constraint. In this section, we prove a near-XP algorithm for $\DC$-Unbounded Subset Sum by first using the constructive Freiman's theorem to map instances to integer programs with small coefficients, then using existing methods to find small-support solutions to the integer programs.
The proof of the lemma uses techniques that are standard in the literature (see,
e.g.,~\cite{DBLP:journals/orl/EisenbrandS06}); nevertheless, we are not aware of a prior proof of the following statement.

\begin{lemma}[ILP Solutions with small support]
    \label{lem:candidates}
    Let $A \in \Z^{m \times n}$ with $\Delta \coloneqq \norm{A}_\infty$. In $(n\Delta)^{O(m)}$ time we can find a set $\mathcal{X}
    \subseteq \{0,1\}^n$ with the following property:
    For any target vector $b \in \Z^m$ corresponding to at least one solution $x \in \ZZ^n$ with $A x = b$, there exists a small-support solution $y \in \ZZ^n$ satisfying
    \begin{displaymath}
        A y = b, \;\; \supp(y) \in \mathcal{X} \;\; \text{ and } \;\; |\supp(y)| \le m \log_2(2n\Delta+1).
    \end{displaymath}
\end{lemma}

\begin{proof}
We begin with a bound on the support of lexicographically minimal solutions that follows standard arguments. 
\begin{claim}
    Let $A \in \Z^{m \times n}$ with $\Delta = \norm{A}_\infty$,  and 
    let $y \in \ZZ^n$ be the lexicographically minimal vector such that $A y = b$
    for some $b \in \Z^m$. Then $|\supp(y)| \le m \log_2 (2n\Delta+1)$.
\end{claim}
\begin{claimproof}
    Assume for contradiction that 
    \[
        2^{|\supp(y)|} > (2 n \Delta+1)^m.
    \]
    Because $Ax \leq (2n\Delta + 1)^m$ for any $x \in \{0,1\}^n$, there must exist two different vectors $v,w \in \{0,1\}^n$ such that (i) $\supp(v),\supp(w) \subseteq
    \supp(y)$, and (ii) $A v = A w$, by the pigeonhole principle. 
    
    Let $y_1 = y - w + v$ and $y_2 = y + w - v$. Observe that $A
    y_1 = A y_2$ and $y_1,y_2 \in \ZZ^n$ because $\supp(v),\supp(w) \subseteq
    \supp(y)$. Moreover, because $v \neq w$ we have that one of $y_1$ or
    $y_2$ is lexicographically smaller than $y$, which contradicts the
    assumption that $y$ is lexicographically minimal.
\end{claimproof}

Let $\mathcal{X} \subseteq \{0,1\}^n$ be the set of lexicographically minimal solutions to $Ax = b$ for every
$b \in \ZZ^n$ with $\norm{b}_\infty < n\Delta$.
Clearly, $|\mathcal{X}| \le (2n\Delta + 1)^m$ as this
is the number of suitable $b$'s. To construct $\mathcal{X}$ it remains to
iterate over every $b \in \ZZ^n$ with $\norm{b}_\infty < n\Delta$ and solve the
following Integer Linear Program:
\begin{displaymath}
    \max \left\{ \sum_{i=1}^n x_i \cdot M^i \mid Ax = b, x \in \ZZ^n\right\},
\end{displaymath}
where $M = 4n\Delta$. Note that this can be solved in $(n\Delta)^{O(m)}$ time
by~\cite[Theorem 2.3]{eisenbrand2019proximity} for each $b$. Hence, the set $\mathcal{X}$ can be constructed in the claimed time.
Finally, it remains to show that for any feasible $b$, there exists a solution $y$ with small support in $\mathcal{X}$.

\begin{claim}
    Let $b \in \Z^m$ be any vector for which there exists $x \in \ZZ^n$ with $A x
    = b$. Then there also exists $y \in \ZZ^n$ such that $Ay = b$ and $\supp(y)
    \in \mathcal{X}$.
\end{claim}
\begin{claimproof}
    Let $z \in \ZZ^n$ be the lexicographically minimum vector such that $Az = b$.
    Let $\hat{z} \in \{0,1\}^n$ be such that $\hat{z}_i = 1$ iff $z_i \neq 0$
    and $\hat{z}_i=0$ otherwise. Let $\hat{b}$ be such that $A\hat{z} =
    \hat{b}$. Observe that $\norm{\hat{b}}_\infty < n\Delta$. 
    
    Hence it remains
    to show that $\hat{z}$ is the lexicographically minimal vector for which
    $A\hat{z} = \hat{b}$. Assume for contradiction that there
    exists $\hat{y} \in \ZZ^n$ such that $A \hat{y} = \hat{b}$ and $\hat{y}$ is lexicographically
    smaller than $\hat{z}$. Consider a vector $y = z - \hat{z} + \hat{y}$. Note,
    that $\supp(\hat{z}) \subseteq \supp(z)$ so $y \in \ZZ^n$. Clearly $Ay = Az
    = b$. Moreover, because
    $\hat{y}$ is lexicographically smaller than $\hat{z}$, it follows that $y$ is
    lexicographically smaller than $z$. This contradicts our assumption that $z$ is lexicographically minimal.
\end{claimproof}
Thus the set $\mathcal{X}$ satisfies the property stated in \Cref{lem:candidates}, concluding the proof.
\end{proof}

With~\Cref{lem:candidates} in hand, let us present our algorithm for $\DC$-Unbounded
Subset Sum.

\begin{theorem}[Near-\XP~algorithm for $\DC$-Unbounded Subset Sum]\label{thm:dc-uss-xp}
    $\DC$-Unbounded Subset Sum can be solved in time $n^{O_\DC(1)}$ if an ILP instance $\mathcal{I}$ on $v$ variables can be solved in time $2^{O(v)} \poly(|\mathcal{I}|)$.

    Using the current best algorithm~\cite{ILP-rothvoss}, $\DC$-Unbounded
    Subset Sum can be solved in $n^{O_\DC(1) \log\log\log(n)}$ time.
\end{theorem}

\begin{proof}
    Following the steps of our reduction from $\DC$-Subset Sum to HBILP feasibility (\Cref{lem:SS-to-HBILP-reduction}), we can use the constructive Freiman's theorem\footnote{We remark that because the construction of the GAP guaranteed by Freiman's theorem is not the runtime bottleneck, a slower constructive algorithm might suffice for this step.} (\Cref{thm:freiman}) to encode the $\DC$-Unbounded Subset Sum instance as an \emph{Unbounded} Hyperplane-Constrained ILP Feasibility instance given by $A \in \ZZ^{d(\DC) \times n}$ with $\Delta = \norm{A}_\infty = n^{O(1 / d(\DC))}$, step vector $\ell$, and target $t$. 

    We then use~\Cref{lem:candidates} to construct a set $\mathcal{X}$ of candidate
    supports in $(n\Delta)^{O(m)} = n^{O_\DC(1)}$ time. For each support vector
    $x^* \in \mathcal{X}$, we reduce the ILP to variables in $x^*$. This gives us a program with $|x^*| = O(m
    \log_2(n\Delta)) = O_{\DC}(\log(n))$ variables. Now, we encode this problem as the ILP
    \begin{displaymath}
        \left\{x \in \ZZ^n \; \middle| \; \sum_{j=1}^d \ell_j \sum_{i \in x^*} a_{i,j} x_i = t \right\}.
    \end{displaymath}

    Observe that this is equivalent to an instance of Unbounded Subset Sum with
    $O_\DC(\log(n))$ items. Thus any algorithm for Unbounded Subset Sum (or, more generally, any algorithm for unbounded ILP) that runs in time $2^{O(v)}$ on instances with $v$ variables would automatically yield an $n^{O_\DC(1)}$ time
    algorithm for $\DC$-Unbounded Subset Sum. Using the best-known algorithm for unbounded ILPs, which runs in time
    $(\log{v})^{O(v)}$ \cite{ILP-rothvoss}, we get an $n^{O_\DC(1)
    \log\log\log(n)}$-time algorithm.
\end{proof}

\section{k-SUM with Constant Doubling}
\label{sec:ksum}

Our final contribution concerns the analogous problem of $k$-SUM with bounded doubling constant, which we refer to as $(\DC, k)$-SUM. We prove 
\Cref{thm:ksum}
and observe that the same approach gives an algorithm for $4$-SUM that is tight up to subpolynomial factors, assuming the $k$-SUM conjecture.

\begin{prob}[$(\DC, k)$-SUM]{prob:kSum}
\textbf{In:} An integer set $X = \{x_1, x_2, \dots, x_n\}$ such that $|X+X|
    \leq \DC|X|$ and an integer target $t$. \\
    \textbf{Out:} $S \subseteq X$ with $|S| = k$ such that $\Sigma(S) = t$, or `NO' if no solution exists.
\end{prob}


We note that \cite{abboud2022stronger} and \cite{jin2022removing} also present algorithms for $3$-SUM in cases where additive structure in the input is controlled by the doubling constant, and also make use of fast algorithms for sparse convolution. In both cases, these authors focus on the setting of tripartite $3$-SUM under the condition that at least one of the three input sets $A$, $B$, and $C$ is guaranteed to have a small doubling.

Before we continue, let us recall the standard "color-coding" technique that allows us to
ensure that each integer in the solution is taken at most once.

\begin{lemma}\label{lem:splitter}
    Let $A$ be a set of $n$ integers. There exists a set family
    \[
        \mathcal{P} \subseteq \{ (A_1,\ldots,A_k) \text{ such that } A_1 \uplus \ldots \uplus A_k \text{ is partition of } A\}
    \]
    with the following properties:
    \begin{enumerate}
        \item For any $S \subseteq A$ of cardinality $|S| = k$, there exists $(A_1,\ldots,A_k) \in \mathcal{P}$ with $|S \cap A_i| = 1$ for all $i \in [k]$.
        \item $|\mathcal{P}| = 2^{O(k)} \cdot \log(n)$.
    \end{enumerate}
     $\mathcal{P}$ can be constructed deterministically in $2^{O(k)} n \log(n)$ time.
\end{lemma}

The proof of~\Cref{lem:splitter} is a reformulation of a standard construction of an $(n,k)$-perfect hash family (see \cite{alon1995color}, Section 4). For completeness, we include a standalone proof in~\Cref{apx:splitter}. We now commence with the proof of~\Cref{thm:ksum}.

\begin{theorem} \label{thm:ksum}
    Given an integer set $X$ such that
    $|X+X| \le \DC \cdot |X|$ and an integer $t$, we can decide if there exists a set $\{x_1,\ldots,x_k\} \subseteq X$ such that
    $x_1 + \ldots + x_k = t$ in deterministic
    time $\tilde{O}(\DC^{\lceil k/2 \rceil} \cdot 2^{O(k)}  \cdot n)$.
\end{theorem}
\begin{proof}
    
Let $X$ be an integer set of size $n$, and let $\{x_1,\ldots,x_k\} \subseteq X$ denote a set of $k$ integers that sum to $t$. We commence by constructing the
family $\mathcal{P}$ from~\Cref{lem:splitter} and guessing a partition
$(X_1,\ldots,X_k) \in \mathcal{P}$ of $X$ such that $x_i \in X_i$ for all $i \in
[k]$. Observe that by~\Cref{lem:splitter} this incurs only an additional $2^{O(k)}
\log(n)$ factor in the running time.

Now, we use the sparse convolution algorithm of Bringmann et al.~\cite{bringmann2022deterministic}.

\begin{lemma}[Theorem 1 in~\cite{bringmann2022deterministic}]\label{lem:sparse-conv}
    Given two sets $A,B \subseteq [\Delta]$, the set $A+B \coloneqq
    \{a+b \mid a \in A, b \in B\}$ can be constructed deterministically in
    $\tilde{O}(|A+B| \cdot \polylog(\Delta))$ time.
\end{lemma}

We use~\Cref{lem:sparse-conv} to enumerate two sets:
\begin{align*}
    \mathcal{L} \coloneqq X_1 + \ldots + X_{\floor{k/2}} &&\text{ and } &&
    \mathcal{R} \coloneqq X_{\floor{k/2}+1} + \ldots + X_k.
\end{align*}


Both $\mathcal{L}$ and $\mathcal{R}$ can be computed
deterministically in
$\tilde{O}(k \cdot (|\mathcal{L}| + |\mathcal{R}|))$ time by
repeatedly applying~\Cref{lem:sparse-conv}.
Next, with both $\mathcal{L}$ and $\mathcal{R}$ in hand, we apply the
meet-in-the-middle approach to recover a solution if one exists. This can be
implemented in $\tilde{O}(|\mathcal{L}| + |\mathcal{R}|)$ time by first sorting
$\mathcal{L}$ and $\mathcal{R}$, and then for every element $a \in \mathcal{L}$ using binary search to decide if $t-a \in \mathcal{R}$. Finally, if for at least
one $a \in \mathcal{L}$ we find an accompanying element in $\mathcal{R}$, we know that the instance has a solution.

As stated, the algorithm decides $k$-SUM without recovering a solution. However, given that a solution exists we can recover a solution via binary search at the cost of an additional $O_k(\log(n))$ factor. This concludes the description of the algorithm.

Correctness of the algorithm follows from the fact that~\Cref{lem:splitter} returns a valid partition, and from the definition of the sets $\mathcal{L}$ and $\mathcal{R}$. It remains to bound the runtime. Since the other steps of the algorithm take time $\tilde{O}_k(n)$, the bottleneck occurs in the meet-in-the-middle step, which takes time $\tilde{O}_k(|\mathcal{L}| +
|\mathcal{R}|)$. Therefore it remains to bound
the sizes of $\mathcal{L}$ and $\mathcal{R}$. Without loss of generality, consider
$|\mathcal{R}|$, and observe
\begin{displaymath}
    |\mathcal{R}| = |X_1 + \ldots + X_{\ceil{k/2}}| \leq |\ceil{k/2}X| = \DC^{\ceil{k/2}} |X|,
\end{displaymath}
where the final step applies Pl{\"u}nnecke-Ruzsa (\Cref{lem:plunnecke}). This concludes the proof of~\Cref{thm:ksum}.
\end{proof}

In the specific case of $k=4$, using the doubling constant directly gives a
slightly better bound. The resulting algorithm for $(\DC, 4)$-SUM is
optimal up to subpolynomial factors under the $4$-SUM conjecture.

\begin{corollary}
    \label{prop:4sum}
    $(\DC, 4)$-SUM can be solved in expected time $\tilde{O}(\DC n)$.
    Moreover, for any constant $\epsilon > 0$, $(\DC, 4)$-SUM cannot be solved
    in $O( \DC^{1-\epsilon} n)$ time unless $4$-SUM can be solved in time $O(n^{2 - \epsilon})$ for $\epsilon > 0$.
\end{corollary}
\begin{proof}
    The upper bound follows by analysis of the proof of~\Cref{thm:ksum}. Recall
    that the bottleneck is $|\mathcal{R}|$, which in the case when $k=4$ is
    $|X+X| \le \DC \cdot |X|$.

    For the lower bound, observe that $|X+X| \le |X|^2$ and therefore $\DC \le
    |X|$. Thus, any algorithm for $(\DC, 4)$-SUM with runtime $O( \DC^{1-\epsilon} n)$ would yield an algorithm for $4$-SUM that runs in $O(n^{2 -
    \epsilon})$ time.
\end{proof}

\begin{remark}
    Applying the same lower bound
    argument to the more general case of $k$-SUM gives a lower bound of
    $\Omega(\DC^{\lceil k/2 \rceil - 1} n)$ for $(\DC, k)$-SUM under the $k$-SUM
    conjecture, leaving an $O(\DC)$-factor gap.
    
    In \cite{petridis2011upper}, Petridis gives the slightly improved bound
    \[
        |hA| = O(\DC^{h-1}|A|^{2 - \frac{1}{h}}),
    \]
    for finite sets in commutative groups, improving on Pl\"unnecke-Ruzsa for our purposes. This
    narrows the gap between our upper and lower bounds slightly, although the result
    is still not tight for $k > 4$. Further improvements, perhaps by non-trivially
    leveraging the small doubling constant of the input set to achieve a better
    algorithmic result for large $k$, would be both interesting and surprising.
\end{remark}

\subsection*{Acknowledgements}
    We thank Lars Rohwedder for insightful discussions that helped to clarify the connections between $\DC$-Subset Sum and Hyperplane-Constrained BILP. We also thank several anonymous reviewers for constructive comments that improved the presentation.

\bibliographystyle{alpha}
\bibliography{main}

\appendix

\section{
    Proof of \texorpdfstring{\Cref{thm:Freiman-constructive}}{}
}
\label{apx:freiman-constructive}

Making Freiman's Theorem constructive is not difficult from a strictly algorithmic perspective. However, verifying the result requires close attention to the structure of the original proof and requires concepts from additive combinatorics, group theory and the geometry of numbers along the way. For this reason, we closely follow Zhao's recent exposition of a proof due to Ruzsa~\cite{ruzsa1994generalized}, making modifications where necessary. We wish to emphasize that \emph{neither the existential results nor the overall proof structure below are novel}. Our contribution is the introduction of algorithmic techniques required to make the proof constructive in near-linear time. 

Given a set $X$ of cardinality $n$, our proof constructs a GAP $P$ of dimension
$d(\DC) = 2^{\DC^{O(1)}}$ and volume $v(\DC) = 2^{2^{\DC^{O(1)}}} n$, as in the
original statement of Freiman's Theorem. We do not attempt to optimize these
functions, but we suspect that techniques used to optimize $d$ and $v$ in
subsequent proofs of Freiman's Theorem (e.g.,
\cite{chang2002polynomial,schoen2011near,sanders2012bogolyubov,sanders2013structure})
could be used to improve the dependence on $\DC$ in our results.

At several points, we make use of the Pl\"{u}nnecke-Ruzsa Inequality, a useful bound on the additive ``growth rate'' of integer sets with a small doubling constant:

\begin{lemma}[Pl\"{u}nnecke-Ruzsa Inequality]
    \label{lem:plunnecke}
    If $X$ is a finite subset of an abelian group and $|X + X| \leq \DC|X|$ for
    a constant $\DC$, then for all nonnegative integers $s$ and $t$,
    \[
        |sX - tX| \leq \DC^{s+t}|X|.
    \]
\end{lemma}

\subsection{Ruzsa's Modeling Lemma}
\label{subsec:ruzsa-modeling}

    A core ingredient in Freiman's theorem is Ruzsa's Modeling Lemma. This allows us to take an integer set $A$ and map a large piece of it to small, finite group (specifically, the prime cyclic group $\Z_q$) in such a way that additive structure is ``preserved'': the image in $\Z_q$ behaves isomorphically to the preimage in $\Z$ under addition, up to a certain fixed number $s$ of additions. The size $q$ of the prime cyclic group is controlled by the size of $|sA - sA|$, which is related to the doubling constant by the Pl\"unnecke-Ruzsa Inequality.

    A map that preserves additive structure in this way is called a \emph{Freiman $s$-isomorphism}:

    \begin{definition}[Freiman Homomorphism and Isomorphism]
        Given subsets $A$ and $B$ of two (possibly different) abelian groups and a positive integer $s \geq 2$, $\phi: A \rightarrow B$ is a \emph{Freiman $s$-homomorphism} if
        \[
                \phi(a_1) + \dots + \phi(a_s) = \phi(a_1') + \dots + \phi(a_s')
        \]
        for all pairs of $s$-tuples in $A$ satisfying $a_1 + \dots + a_s = a_1' + \dots + a_s'$. $\phi$ is a \emph{Freiman $s$-isomorphism} if $\phi$ is a bijection and both $\phi$ and $\phi^{-1}$ are Freiman $s$-homomorphisms.
    \end{definition}

    \begin{lemma}[Constructive Ruzsa's Modeling Lemma, c.f. \cite{zhao2022graph} Theorem 7.7.3]
        \label{lem:RML-constructive}
        Let $A$ be a set of $n$ integers with $|A+A| \leq \DC |A|$, set $\Delta = \max_{a \in A} |a|$, let $s \geq 2$ be a fixed
        constant, and set $m = 4 \DC^{2s} n$. There exists an $O(n + \polylog(\Delta))$-time algorithm that:
        \begin{enumerate}
            \item with probability at least $1/2$, returns a mapping $\psi: \Z \rightarrow \Z_m$ and a set $A' \subset A$ with $|A'| \geq |A|/s$ such that $\psi$ is a $s$-Freiman isomorphism from $A'$ to $\psi(A')$, and
            \item with probability at most $1/2$, returns `failure'.
        \end{enumerate}
    \end{lemma}

    
    \begin{proof}
        Fix any prime $q > \max(sA - sA)$. As $|\max(sA - sA)| \leq s\Delta$, we can
        find a prime of this size with high probability in time $O_s(\polylog(\Delta))$ by repeatedly guessing and testing primality \cite{agrawal2004primes}.
        
         For each value $\lambda \in [q]$, let $\phi_\lambda$ denote the map $\phi_\lambda: \Z \rightarrow \Z_q \rightarrow \Z_q \rightarrow \Z$ that 
        \begin{enumerate}
            \item first maps $a \in \Z$ to $a' = a \pmod{q} \in \Z_q$,
            \item computes $a'' = \lambda a'$ in $\Z_q$,
            \item then maps $a''$ back to $[0:q-1] \subset \Z$ via the identity map.
        \end{enumerate}

        Choose $\boldsymbol{\lambda} \in [q - 1]$ uniformly at random. Since $q$
        is prime, any element $r \in [q-1]$ is a generator for $\Z_q$, and
        thus for any $r \in [q-1]$, $\phi_{\boldsymbol{\lambda}}(r)$ is
        uniformly distributed over $[q-1]$. Since $q > \max(sA - sA)$, for any nonzero integer $c \in sA - sA$, $c \in [q-1]$ and thus $\phi_{\boldsymbol{\lambda}}(c)$ is uniformly random over $[q-1]$. 

        Thus for any nonzero $c \in sA - sA$ the probability that $\phi_{\boldsymbol{\lambda}}(c)$ is divisible by $m$ is less than $2/m$. As $m = 4\DC^{2s} n \geq 4|sA - sA|$ by \Cref{lem:plunnecke}, $|sA - sA| \leq m/4$ and the probability that $m$ evenly divides \emph{any} nonzero element in $sA - sA$ is less than $1/2$ by a union bound. Compute $\phi_{\boldsymbol{\lambda}}(A)$ in time $O(n)$ and output ``failure'' if $m$ divides any element in this set. Otherwise, continue.

        Let $A'$ be a subset of $A$ such that $|A'| \geq n / s$ and
        $\mathrm{diam}(\phi_{\boldsymbol{\lambda}}(A')) \leq q / s$. Note that the existence of $A'$ is guaranteed by the pigeonhole principle. We can compute $A'$ in time $O(n)$ by partitioning $[q]$ into evenly-sized intervals and computing $\phi_{\boldsymbol{\lambda}}(A)$. 
        
        Finally, we define $\psi_{\boldsymbol{\lambda}}: \Z \rightarrow \Z_m$ by $\psi_{\boldsymbol{\lambda}}(x) = \phi_{\boldsymbol{\lambda}}(x) \pmod{m}$ and observe that $\psi_{\boldsymbol{\lambda}}$ is a $s$-isomorphism from $A'$ to $\psi_{\boldsymbol{\lambda}}(A')$ as $m$ does not divide any nonzero element in $sA - sA$. This follows from the final two paragraphs of the proof of Theorem 7.7.3 in \cite{zhao2022graph}, with the argument unchanged.
    \end{proof}

\subsection{Bogolyubov's Lemma in \texorpdfstring{$\Z_m$}{} }
    Given a relatively large set $B \in \Z_m$, Bogolyubov's Lemma states that $2B - 2B$ contains a set of points that behaves ``like a subspace'' in the sense that each point is ``close to orthogonal'' to a certain set $R \in \Z_m$. Specifically, we employ the concept of a \emph{Bohr set}, defined as
    \[
            \text{Bohr}_m(R, \epsilon) \coloneqq \{x \in \Z_m:
            \norm{rx / m}_{\mathbb{R}/\Z} \leq \epsilon, \text{ for all } r \in R\}.
    \]
    (Recall that the norm $\norm{\cdot}_{\mathbb{R}/\Z}$ denotes distance from the nearest integer.) We refer to $|R|$ as the \emph{dimension} and $\epsilon$ as the \emph{width} of the Bohr set. 

    A Bohr set is analogous to a subspace of codimension $|R|$, in the sense that if we add together several elements of a Bohr set, their sum is still close to a multiple of $m$ when scaled by any $r \in R$. In this sense, we can view Bogolyubov's lemma as a statement that sets of the form $2B - 2B \in \Z_m$ contain subsets with group-like structure.

    \begin{lemma}[Constructive Bogolyubov's lemma for $\Z_m$, c.f. \cite{zhao2022graph} Theorem 7.8.5]
        \label{lem:BL-constructive}
        Given $B \subseteq \Z_m$ with $|B| = \alpha m$, we can
        compute $R \subseteq \Z_m$ of dimension $|R| <
        1/\alpha^2$ such that $\text{Bohr}(R, 1/4) \subseteq 2B - 2B$ in time $\tilde{O}(m)$.
    \end{lemma}
    \begin{proof}
    To make Bogolyubov's lemma in $\Z_m$ \cite[Theorem 7.8.5]{zhao2022graph} constructive, it suffices to observe that $R$ is defined explicitly as 
        \[
            R = \{ r \in \Z_m \setminus \{0\} \; : \;
            |\hat{\mathbbm{1}_B}(r)| > \alpha^{3/2}\}.
        \]
        Here $\hat{\mathbbm{1}_B}$ is the finite group Fourier transform of
        $\mathbbm{1}_B$, the membership function of $B$: 
        \[
            \hat{\mathbbm{1}_B}(r) = \frac{1}{m} \sum_{x \in
            \Z_m} \mathbbm{1}_B(x) e^{-\frac{2\pi i r x}{m}}.
        \]
        Computing $R$ directly using the Fast Fourier Transform takes time $O(m \log m)$.
    \end{proof}

\subsection{ Finding a GAP in a Bohr Set }

    The structured nature of the Bohr set is instrumental in constructing a generalized arithmetic progression: in fact, we can show that every Bohr set contains a large GAP. In order to prove this, we need to introduce definitions from the geometry of numbers.

    \begin{definition}[Successive Minima and Directional Basis]
        Let $\Lambda \subseteq \mathbb{R}^d$ be a lattice and $T \subseteq \mathbb{R}^d$ be a centrally symmetric convex body.
    
        For $i \in [d]$, the \emph{$i$th successive minimum} $\lambda_i$ of $T$
        with respect to $\Lambda$ is the minimum value such that $\Lambda \cap
        \lambda_i \cdot T$ contains $i$ linearly independent lattice vectors.

        A \emph{directional basis} of $T$ with respect to $\Lambda$ is a basis $\{b_1, b_2, \dots, b_d\}$ of $\mathbb{R}^d$ such that for each $i \in [d]$, $b_i \in \lambda_i T$.
    \end{definition}

    In visual terms, we can imagine constructing a directional basis by gradually scaling the convex body $T$ outward from the origin. Every time $T$ engulfs a new lattice vector $v$, we add $v$ to our directional basis if and only if $v$ is linearly independent from the current set of basis vectors.

    \begin{lemma}[Constructing a Large GAP in a Bohr Set, c.f. \cite{zhao2022graph} Theorem 7.10.1]  
        \label{lem:GAPinBohr-constructive}
        Let $m$ be a prime. Given a set $R \subseteq \Z_m$ of size $|R| = d$ and $\epsilon \in (0,1)$, we can compute a proper GAP $P \subseteq Bohr(R, \epsilon)$ with dimension at most $d$ and volume at least $(\epsilon / d)^d m$ in time $\tilde{O}_d(m)$.
    \end{lemma}
    \begin{proof}
        Let $R = \{r_1, r_2, \dots, r_d\}$ be a subset of $\Z_m$ (recall
        that $m$ is a prime). We can directly compute the vector $v = (\frac{r_1}{m}, \dots, \frac{r_d}{m}) \in \mathbb{R}^d$ to define the lattice
        \[
        \Lambda = \Z^d + \Z v \subseteq \mathbb{R}^d.
        \]
        Note that the lattice vectors are not necessarily integral, and we have not yet computed a lattice basis; letting $e_i$ denote the standard basis vector in dimension $i$, the set $\{e_1, e_2, \dots, e_d, v\}$ spans the lattice but is not linearly independent.

        Let $r'$ be any nonzero element of $R$. Since $\Z_m$ is a cyclic group of prime order, $r'$ generates $\Z_m$. Thus, because one component of $v$ is $r' / m$, the translations of the integer lattice $\{\Z^d + \gamma v\}_{\gamma \in [m]}$ are all disjoint. Since
        \[
            \Lambda = \Z^d + \Z v \subseteq \bigcup_{\gamma \in [m]} \Z^d + \gamma v,
        \]
        we have that $\Lambda$ is the disjoint union of $m$ translates of the integer lattice. This implies that there are exactly $m$ lattice points of $\Lambda$ within each translation of the unit cube, and, equivalently, that $\det(\Lambda) = 1/m$.
        
        As a result, we can enumerate the set
        \[
            C \coloneqq \{ \norm{ \gamma r / m }_{\R \setminus \Z} \; : \; \gamma \in [m]\} - \{0,1\}^d,
        \]
        the set of all $2^d \cdot m$ lattice points in the cube $\Lambda \cap [-1, 1)^d$, in time $O(2^d \cdot m) = O_d(m)$.

        Next, we sort the set $C$ according to the $L_\infty$ metric, which
        takes time $\tilde{O}_d(m)$. This coincides with our definition of the
        successive minima of a cube centered on the origin with respect to
        $\Lambda$: if $\lambda_i$ is the $i$th successive minima of $[-\epsilon,
        \epsilon]^d$ with respect to $\Lambda$, then the $i$th directional basis
        vector satisfies $\norm{b_i}_\infty \leq \lambda_i [-\epsilon, \epsilon]^d$.
        
        Construct the successive minima $\lambda_1, \dots, \lambda_d$ and the directional basis $b_1, \dots, b_d$ of $[-\epsilon, \epsilon]^d$ with respect to $\Lambda$ by greedily adding independent lattice vectors to our basis from short to long according to the $L_\infty$ metric. Checking whether each subsequent lattice vector is independent from the previous set takes time $O_d(1)$ using Gaussian elimination. Because $[-1, 1)^d$ contains $d$ linearly independent lattice vectors (consider the standard basis), our directional basis is guaranteed to be contained in $\Lambda \cap [-1, 1)^d = C$.

        To complete the construction of the GAP in \cite[Theorem 7.10.1]{zhao2022graph}, we observe that the proper GAP $P$ is defined explicitly in terms of the directional basis of $[-\epsilon, \epsilon]^d$ with respect to $\Lambda$ that we have just constructed. Specifically, we have
        \[
            P = \{ \ell_1 x_1 + \cdots + \ell_d x_d \; | \; \forall i \in [d], \ell_i \in [L_i] \},
        \]
        where each $x_i$ is the unique element in $[0:m-1]$ such that $b_i \in
        x_i v + \Z^d$ and $L_i \coloneqq \lceil 1/(\lambda_i d)\rceil$. Each $L_i$ can be computed directly, and each $x_i$ can be computed in time $O(m)$.
    \end{proof}

\subsection{Ruzsa's Covering Lemma}
\label{subsec:ruzsa-covering}

    Ruzsa's covering lemma states that if the sumset $|Y+Z|$ is small relative to $|Y|$, it is possible to cover $Z$ with a small number of translates of $Y - Y$. A rough intuition for this result is that it is a statement about the ``conservation of additive structure'': if $Y$ and $Z$ have ``common additive structure'' (captured by the condition that $|Y+Z| \leq \DC |Y|$), then $Z$ is ``similar'' to $Y - Y$ (in the sense that $Z$ is covered by few translates of $Y - Y$).
    
    The fact that Ruzsa's covering lemma can be made efficiently constructive
    was previously observed by Abboud, Bringmann, and Fischer~\cite{abboud2022stronger}:

    \begin{lemma}[Constructive Ruzsa's Covering Lemma, \cite{abboud2022stronger} Lemma 4.7]
        \label{lem:RCL-constructive}
        Let $Y, Z$ be nonempty finite subsets of an abelian group. If $|Y+Z|
        \leq \DC |Y|$, then there exists a subset $X \subseteq Z$ with $|X| \leq
        \DC$ and $Z \subseteq Y - Y + X$. Moreover, $X$ can be computed in time
        $\tilde{O}(\frac{|Y - Y + Z|\cdot |Y+Z|}{|Y|}) = \tilde{O}(\DC|Y - Y + Z|)$.
    \end{lemma}


\subsection{ Proof of \texorpdfstring{\Cref{thm:Freiman-constructive}}{} }

\begin{proof}
    Combining the ingredients from the previous subsections allows us to prove
    \Cref{thm:Freiman-constructive}. Let $A$ be a finite
    integer set with $|A+A| \leq \DC |A| = \DC n$. By \Cref{lem:plunnecke}, $|8A
    - 8A| \leq \DC^{16} |A|$. 

    Choose a prime $m = O_\DC(n)$ satisfying $4 \DC^{16} n < m < 16 \DC^{16} n$,
    which can be accomplished in time $O_\DC(\polylog(n))$ with high probability by guessing and testing primality \cite{agrawal2004primes}. Then, apply
    \Cref{lem:RML-constructive} with $s = 8$ to compute a set $A' \subseteq A$
    with $|A'| \geq |A| / 8$ and a mapping $\psi$ such that $\psi$ is a Freiman
    8-isomorphism from $A'$ to $B \coloneqq \psi(A') \subseteq \Z_m$ with probability at least $1/2$ in time $\tilde{O}(n)$. We can increase the success probability of this step by repetition: for any integer constant $\gamma > 0$, running the algorithm $\gamma \log(n)$ times lowers the failure probability to $n^{-\gamma}$.

    Apply \Cref{lem:BL-constructive} to $B$ with 
    \[
    \alpha = \frac{|B|}{m} = \frac{|A'|}{m} \geq \frac{|A|}{8m} = 1 / O_\DC(1).
    \]
    This gives us $R \subseteq \Z_m$ of size $1 / \alpha^2 =
    O_\DC(1)$ such that $Bohr(R, 1/4) \subseteq 2B - 2B$ in time $\tilde{O}(m) =
    \tilde{O}(n)$. Then, apply \Cref{lem:GAPinBohr-constructive} to $R$ to
    compute the proper GAP $P \subset Bohr(R, \epsilon) \subseteq 2B-2B$ with
    dimension at most $|R| = O_\DC(1)$ and volume at least $(1/4|R|)^{|R|} m = m
    / O_\DC(1)$.

    Following the proof of Theorem 7.11.1 (Freiman's Theorem) in
    \cite{zhao2022graph}, we have that $Q \coloneqq \psi^{-1}(P)$ is a GAP of the same dimension and volume satisfying 
    \begin{equation}
        \label{eq:QA1}
        Q \subseteq 2A' - 2A' \subseteq 2A - 2A.
    \end{equation} 
    Thus
    \[
        |Q + A| \leq |2A - 2A + A| = |3A - 2A| \leq \DC^5 |A| = O_\DC(1) \cdot |Q|,
    \]
    where the second inequality uses \Cref{lem:plunnecke}. Using the fact that $|Q+A| = O_\DC(1) \cdot |Q|$, we apply \Cref{lem:RCL-constructive} to $Q$ and $A$ to get a set $X$ of size $|X| = O_\DC(1)$ satisfying $A \subseteq Q - Q + X$ in time 
    \[
        \tilde{O}(\DC|Q - Q + A|) = \tilde{O}(\DC|2A - 2A - (2A - 2A) + A|) =
        \tilde{O}(\DC|5A - 4A|) = \tilde{O}(\DC^9|A|),
    \]
    where the final equality uses \Cref{lem:plunnecke}. We conclude with the
    observation that $Q - Q + X$ is a GAP of dimension $O_\DC(1)$ and volume $O_\DC(1) \cdot |Q| = O_\DC(1) \cdot |A|$, containing $A$. Note that each step in the proof takes $\tilde{O}_\DC(n)$ time.
\end{proof}

\subsection{Proof of \texorpdfstring{\Cref{obs:bigL-bound}}{}}
\label{apxsec:bigL-bound}

\begin{proof}
    Suppose $L_i > n^{1/d(\DC)}$ for some $i \in [d(\DC)]$, and let $\alpha_i
    \coloneqq \alpha_i(\DC)$ be the solution to $L_i \leq n^{\alpha_i/d(\DC)}$. (Note that $\alpha_i = O_\DC(1)$, as $|A| = O_\DC(n)$.) 
        
    Let $\hat{\alpha}_i \coloneqq \alpha_i - \lfloor \alpha_i \rfloor$ denote the decimal part of $\alpha_i$, and observe that
    \begin{align}
        &\{y_i\ell_i : \ell_i \in [L_i]\} \subseteq \\
        &\{y_{i,1}\ell_{i, 1} + \dots  +y_{i,\lfloor \alpha_i \rfloor}\ell_{i, \lfloor \alpha_i \rfloor} + y_{i,\lceil \alpha_i \rceil}\ell_{i, \lceil \alpha_i \rceil} \; : \; \forall j \in [\lfloor \alpha_i \rfloor], \ell_{i, j} \in [\lceil L_i^{1/d(\DC)} \rceil], \ell_{i, \lceil \alpha_i \rceil} \in [n^{\hat{\alpha}_i/d(\DC)}]\},
    \end{align}
    where $y_{i, j} = y_i \lceil L_i^{1/d(\DC)} \rceil^{j-1}$ for $j \in [\lceil \alpha_i \rceil]$; that is, we can replace one dimension of our arithmetic progression with $\lceil\alpha_i\rceil$ new dimensions, each bounded by $n^{1/d(\DC)}$. As
    \[
        vol(P) = \prod_{i \in [d]} L_i = O_\DC(n)
    \]
    by \Cref{thm:Freiman-constructive}, performing this operation for each $L_i > n^{1/d(\DC)}$ results in a new gap $P'$ with dimension $d'(\DC) \leq 2d(\DC)$ and volume $O_\DC(n)$. 
\end{proof}

\section{ ILP Manipulations }
\label{apx:ilp-manipulation}

This appendix contains manipulations that allow us to assume certain properties of ILPs without loss of generality.

\subsection{ Non-negativity for BILP Feasibility }
\label{apxsubsec:BILP-non-negative}

\begin{observation}\label{obs:non-negative}
    Let $\mathcal{I}$ be an instance of ILP feasibility with binary variables given by the constraint matrix $A \in \Z^{m \times n}$ and the target vector $b \in \Z^m$. Without loss of generality, we can assume that entries of $A$ are non-negative and that every solution has fixed support $q$ for some $q = \Theta(n)$.
\end{observation}
\begin{proof}

Construct a new constraint matrix $\tilde{A}$ as follows: Recall that $J_{m\times n}$ denotes the $m$ by $n$ matrix of 1's, and add $\Delta \cdot J_{m\times n}$ to the matrix $A$. Then we append an additional $n$ columns to $A$, where each column consists only of $\Delta$ entries only. Finally, we append a row of 1's. 

To create $\tilde{b}$, add $n\Delta \cdot J_{m \times 1}$ to $b$ and append a single entry with the value $n$. 

\[
    \tilde{A} \coloneqq 
    \begin{pmatrix}
        \makebox(2.5cm,1.5cm){ $A + \Delta \cdot J_{m \times n}$}
        & \rvline & \makebox(2.5cm,1.5cm){$\Delta \cdot J_{m \times n}$} \\
        \hline
        \makebox(1.5cm,0.5cm){ $J_{1 \times n}$} & \rvline &
        \makebox(1.5cm,0.5cm){ $J_{1 \times n}$}
    \end{pmatrix}
    \smallskip
    \,\,\,
    \tilde{b} \coloneqq
    \begin{pmatrix}
        \makebox(2.3cm,1.5cm){ $b + n\Delta \cdot J_{m \times 1}$} \\
        \hline
        \makebox(1cm,0.5cm){$n$} \\
    \end{pmatrix}
\]

Observe that every entry in $\tilde{A}$ is positive and that the maximum entry in $\tilde{A}$ is at most $2\Delta = O(\Delta)$. For correctness, note that the last row ensures that any solution to $\tilde{A} x = \tilde{b}$ with $x \in \{0,1\}^{2n}$ has support exactly $n$. This implies that the additional $\Delta$ factors added to every component in each of the first $m$ rows add a total of $n\Delta$ to each component of $Ax$. Thus $\tilde{A}x = \tilde{b}$ if and only if $Ax = b$.
\end{proof}

\subsection{ Non-negativity for HBILP Feasibility }
\label{apxsubsec:HBILP-non-negative}

\begin{observation}\label{obs:HBILP-non-negative}
    Let $\mathcal{I}$ be an instance of HBILP feasibility given by the constraint matrix $A \in \Z^{m \times n}$, the step vector $s \in \Z^m$, and the target $t \in \Z$. Without loss of generality, we can assume that entries of $A$ are non-negative and that every solution $x$ has fixed support size $q$ for some $q = \Theta(n)$.
\end{observation}
\begin{proof}
Given $A, s, t$, we create a new HBILP feasibility instance $\tilde{A}, \tilde{s}, \tilde{t}$ as follows. Recall that $J_{m \times n}$ denotes the $m$ by $n$ matrix of 1's and add $\Delta \cdot J_{m \times n}$ to $A$. Then append the matrix $\Delta \cdot J_{m\times n}$ to the right-hand side $A$. Finally, add a row of 1's to the bottom of the matrix. 

Define 
\[
    M \coloneqq \norm{s}_\infty \cdot 3n\Delta m,
\]
and note that, by construction, we have
\begin{equation}
    \label{eq:M-lb-HBILP}
    \langle (A + 2 \Delta J_{m \times n})x, s \rangle \leq 3n\Delta \langle J_{m \times 1}, s \rangle < M.
\end{equation}
Create $\tilde{s} \in \Z^{m+1}$ by appending $M$ to $s$, and set $\tilde{t} = t
+ n\Delta\norm{s}_1 + nM$. Written as block matrices, we have:

\[
    \tilde{A} \coloneqq 
    \begin{pmatrix}
        \makebox(2.5cm,1.5cm){ $A + \Delta \cdot J_{m \times n}$}
        & \rvline & \makebox(2.5cm,1.5cm){$\Delta \cdot J_{m \times n}$} \\
        \hline
        \makebox(1.5cm,0.5cm){ $J_{1 \times n}$} & \rvline &
        \makebox(1.5cm,0.5cm){ $J_{1 \times n}$}
    \end{pmatrix}
    \smallskip
    \,\,\,
    \tilde{s} \coloneqq
    \begin{pmatrix}
        \makebox(2.3cm,1.5cm){ $s$ } \\
        \hline
        \makebox(1cm,0.5cm){$M$} \\
    \end{pmatrix}
\]

Observe that every entry in $\tilde{A}$ is positive and that the maximum entry in $\tilde{A}$ is at most $2\Delta = O(\Delta)$. Because the top $m$ rows of $\tilde{A}$ contribute a total value less than $M$ to the dot product $\langle \tilde{A}x, \tilde{s} \rangle$ by \eqref{eq:M-lb-HBILP}, any solution to $\tilde{A}, \tilde{s}, \tilde{t}$ must have support exactly $n$ so that the resulting dot product contains the term $nM$.

It remains to prove correctness:
\begin{claim}
    A vector $y \in \{0, 1\}^{2n}$ satisfies $\langle \tilde{A}y, \tilde{s} \rangle = \tilde{t}$ if and only if $\supp(y) = n$ and the first half of $y$, the vector $y' = (y_1, y_2, \dots, y_n)$, satisfies $\langle Ay', s \rangle = t$.
\end{claim}
\begin{claimproof}
    Suppose some vector $y' \in \{0, 1\}^n$ satisfies $\langle Ay', s \rangle =
    t$. Then the vector $y$ created by adding an arbitrary $n$-bit string with support $n - \supp(y')$ satisfies
    \[
        \langle \tilde{A}x', \tilde{s} \rangle = \tilde{t}
    \]
    and is a valid solution to $\tilde{A}, \tilde{s}, \tilde{t}$.

    Now suppose some vector $y \in \{0, 1\}^{2n}$ satisfies $\langle \tilde{A}y, \tilde{s} \rangle = \tilde{t}$. As previously noted, we must have $\supp(y) = n$ to create the $nM$ term in the product $\langle \tilde{A}y, \tilde{s} \rangle$. The additional $\Delta$ factors added to every component in each of the first $m$ rows of $\tilde{A}$ create the $n \Delta ||s||_1$ term in the product. If we remove these two terms, the remainder of the equation $\langle \tilde{A}y, \tilde{s} \rangle = \tilde{t}$ is $\langle A(y'_1, y'_2, \dots, y'_n), s \rangle = t$.
\end{claimproof}

    This concludes the proof of \Cref{obs:HBILP-non-negative}.
\end{proof}

\subsection{ Reduction of Binary ILP Feasibility to HBILP Feasibility }
\label{apxsubsec:BILP-to-HBILP}

\begin{proof}[Proof of \Cref{lem:BILP-to-HBILP}]
    Fix an instance $A \in \Z^{m \times n}, b \in \Z^n$ of Binary ILP
    Feasibility with $\Delta \coloneqq \norm{A}_\infty$. By \Cref{obs:non-negative}, we can assume without loss of generality that the entries of $A$ are non-negative.

    Define $q \coloneqq q(n, \Delta) = n\Delta + 1$ and create a new instance $A, s, t'$ of HBILP feasibility by setting
    \begin{align*}
        s &\coloneqq (q^0, q^1, \dots, q^{m-1}) \text{ and } \\
        t' &\coloneqq \langle b, s \rangle = b_1 \cdot q^0 + b_2 \cdot q^1 + \dots + b_m \cdot q^{m-1},
    \end{align*}
    effectively using $t'$ to store $m$ registers of $\log_2(q) =
    \log_2(n\Delta+1)$ bits each. 
    
    We claim $x \in \{0,1\}^n$ solves $A, b$ if and only if it solves $A, s, t'$. If $Ax = b$, $\langle Ax, s \rangle = t'$ follows immediately from the definition of $t'$. 
    
    Now suppose $\langle Ax, s \rangle = t'$. Because $q > A[i, \cdot]x$ for any row $i \in [m]$ by construction, the only way to achieve $t'$ is if $A[i, \cdot]x = b_i$ for each $i \in [m]$.
\end{proof}

\section{Proof of~\texorpdfstring{\Cref{lem:splitter}}{}}\label{apx:splitter}

The proof of~\Cref{lem:splitter} reformulates the well-known concepts of the \emph{perfect hash family} and the \emph{splitter}. 
\begin{definition}[Splitter] \label{def:splitter}
    An $(n,k,\ell)$-splitter $\mathcal{F}$ is a family
    of functions from $[n]$ to $[\ell]$ such that for every set $S \subseteq [n]$
        of size $k$, there exists $f \in \mathcal{F}$ such that for every $1 \le
    j,j' \le \ell$, the values $|f^{-1}(j) \cap S|$ and $|f^{-1}(j') \cap S|$
    differ by at most $1$.
\end{definition}

In other words, for every $S \subseteq [n]$ of size $k$, some $f \in \mathcal{F}$ partitions $[n]$ into $\ell$ subsets in a way that splits $S$ as evenly as possible.
The special case of an $(n,k,k)$-splitter is called $(n,k)$-perfect hash family.
We use the following construction of an $(n,k)$-perfect hash family due to Naor et
al.~\cite{naor1995splitters}.

\begin{theorem}[\cite{naor1995splitters}]\label{thm:perfect-hash}
    For any $n,k \ge 1$, it is possible to construct an $(n,k)$-perfect hash family of
    size $e^k k^{O(\log k)} \log n$ in time $e^k k^{O(\log(k)} n \log n$.
\end{theorem}

Observe that in our case, the $k^{O(\log k)}$ factor is absorbed by the $2^{O(k)}$ factor in the statement of~\Cref{lem:splitter}.  Let $A = \{a_1,\ldots,a_n\}$. For each function $f \in \mathcal{F}$ and each integer
$i \in [k]$ we let $f_A : [k] \to 2^A$ be

\begin{displaymath}
    f_A(i) = \{ a_j \mid f(j) = i \}.
\end{displaymath}

With the perfect hash family $\mathcal{F}$ we construct
the set $\mathcal{P}$ as follows: for every function $f \in \mathcal{F}$ we
simply add the set family $(f_A(1),\ldots,f_A(k))$
to the set $\mathcal{P}$. 
Observe that this set family forms a partition of $A$ because $f$ is a
well-defined function.  \Cref{thm:perfect-hash} provides the claimed guarantees
on the size of $\mathcal{P}$ and the construction time. Finally, let $S =
\{x_1,\ldots,x_k\}$ be an arbitrary subset of $A$. Because $\ell = k$, ~\Cref{def:splitter} guarantees that for some $f \in \mathcal{F}$, $|f^{-1}(j)\cap S| = 1$ for every $j \in
[k]$. This concludes the proof of~\Cref{lem:splitter}.

\end{document}